\pdfoutput=1

\documentclass{article}

\usepackage[noend,ruled, linesnumbered]{algorithm2e} 
\usepackage{arxiv}
\usepackage{mathtools}

\usepackage{xcolor}
\usepackage{graphicx, caption, subcaption}
\usepackage{booktabs} 
\usepackage{fancyvrb} 
\usepackage{makecell} 
\usepackage{multirow} 
\usepackage{tabulary}
\usepackage{array}
\usepackage{enumitem}
\usepackage{hyperref}
\usepackage{cleveref}
\usepackage{amsthm}
\usepackage{amssymb}

\usepackage[utf8]{inputenc}
\usepackage[english]{babel}

\PassOptionsToPackage{dvipsnames}{xcolor}
\PassOptionsToPackage{dvipsnames, xcdraw, table}{xcolor}

\newcolumntype{L}[1]{>{\raggedright\let\newline\\\arraybackslash\hspace{0pt}}m{#1}}
\newcolumntype{C}[1]{>{\centering\let\newline\\\arraybackslash\hspace{0pt}}m{#1}}
\newcolumntype{R}[1]{>{\raggedleft\let\newline\\\arraybackslash\hspace{0pt}}m{#1}}
\usepackage{xspace} 

\SetKwRepeat{Do}{do}{while} 

\newcommand{\rpc}[1]{}
\newcommand{\trung}[1]{}
\newcommand{\dumi}[1]{}

\newcommand{\revision}[1]{{#1}}

\newcommand{\fabricSharp}{FabricSharp}
\newcommand{\ffabricSharp}{FastFabricSharp}
\newcommand{\fabricPlusplus}{Fabric++}

\newcommand{\entry}[3]{\ensuremath{(\texttt{#1}, (#2), #3)}}
\newcommand{\dbentry}[4]{\ensuremath{\{{\texttt{#1}}\_{#2}\_{#3}: \texttt{#4}\}}}

\newcommand{\startTS}[1]{\ensuremath{\texttt{StartTs}(\texttt{#1})}}
\newcommand{\commitTS}[1]{\ensuremath{\texttt{EndTs}(\texttt{#1})}}

\newcommand{\scenario}[1]{\textit{#1}}

\def\unioneq{\ensuremath{\mathrel{{\cup}{=}}}}

\def\xmark{\ensuremath{\texttt{x}}}
\def\vmark{\ensuremath{\checkmark}}

\newtheorem{theorem}{Theorem}
\newtheorem{proposition}{Proposition}
\newtheorem{definition}{Definition}
\newtheorem{lemma}{Lemma}

\author{
  Pingcheng Ruan \\
  National University of Singapore \\
  \texttt{ruanpc@comp.nus.edu.sg} \\
   \And
  Dumitrel Loghin \\
  National University of Singapore \\
  \texttt{dumitrel@comp.nus.edu.sg} \\
   \And
  Quang-Trung Ta \\
  National University of Singapore \\
  \texttt{taqt@comp.nus.edu.sg} \\
   \And
  Meihui Zhang \\
  Beijing Institute of Technology \\
  \texttt{meihui\_zhang@bit.edu.cn} \\
   \And
  Gang Chen \\
  Zhejiang University \\
  \texttt{cg@zju.edu.cn} \\
   \And
  Beng Chin Ooi \\
  National University of Singapore \\
  \texttt{ooibc@comp.nus.edu.sg} 
}

\title{A Transactional Perspective on Execute-order-validate Blockchains}

\begin{document}
\maketitle

\begin{abstract}

Smart contracts have enabled blockchain systems to evolve from simple
cryptocurrency platforms, such as Bitcoin, to general transactional systems,
such as Ethereum.
Catering for emerging business requirements, a new architecture called
execute-order-validate has been proposed in Hyperledger Fabric to support
parallel transactions and improve the blockchain's throughput.
However, this new architecture might render many invalid transactions when
serializing them.
This problem is further exaggerated as the block formation rate is inherently
limited due to other factors beside data processing, such as
cryptography and consensus.

In this work, we propose a novel method to enhance the execute-order-validate
architecture, by reducing invalid transactions to improve the throughput of
blockchains.
Our method is inspired by state-of-the-art optimistic concurrency control
techniques in modern database systems.
In contrast to existing blockchains that adopt database's preventive approaches
which might abort serializable transactions,
our method is theoretically more fine-grained. Specifically, unserializable transactions
are aborted before ordering and the remaining transactions are guaranteed to be
serializable.
\revision{For evaluation, we implement our method in two blockchains respectively, {\fabricSharp} on top of
Hyperledger Fabric, and {\ffabricSharp} on top of FastFabric. 
}
We compare the performance of {\fabricSharp} with vanilla Fabric and three related systems, two of which are respectively implemented with one standard and one state-of-the-art concurrency control techniques from databases.
The results demonstrate that {\fabricSharp} achieves \revision{25\%} higher throughput compared to the other systems in nearly all experimental scenarios.
\revision{Moreover, the {\ffabricSharp}'s improvement over FastFabric is up to 66\%. 
}
%


\end{abstract}



\section{Introduction}
\label{sec:intro}

Blockchains have stricken the world like a storm.
The concept of \textit{blockchain}, originating from Nakamoto's Bitcoin
whitepaper \cite{nakamoto2008bitcoin}, proposes to employ a hashed chain of
blocks to batch historical monetary transactions.
This chain is distributed across a network of mutually distrusting nodes that
run a \textit{proof-of-work} (PoW) consensus to consistently replicate the chain
and synchronize the state.
The consensus groups the transactions into blocks, and the nodes serially
execute the transactions in a block to update their state.
While Bitcoin only supports cryptocurrency operations, Ethereum was designed to
support \emph{Turing-complete} smart contracts that encode \textit{arbitrary}
data processing logic \cite{wood2014ethereum}.
With Ethereum, blockchains evolved from cryptocurrency platforms to distributed
transactional systems.

Blockchain systems can be classified into \textit{permissionless}
(\textit{public}), such as Bitcoin and Ethereum, and \textit{permissioned}
(\textit{private}), such as Hyperledger Fabric~\cite{androulaki2018hyperledger}.
In public blockchains, the data and transactional logic are transparent to the
public, hence, are subject to private data leakage.
Due to their openness, public blockchains use expensive PoW consensus.
This, together with the serial transaction execution limit these systems' capacity.
Addressing the limitations of public blockchains, Hyperledger Fabric is a
private blockchain that supports \emph{concurrent} transactions
\cite{androulaki2018hyperledger}.
A Fabric blockchain requires its members to enroll through a trusted membership service in order to interact with the blockchain.
In this paper, we focus on permissioned blockchains as they are more suitable
for supporting applications such as supply-chain, healthcare and resource
sharing, and in particular, we use Fabric as the underlying blockchain system.

Fabric supports a new transaction execution architecture called
execute-order-validate (EOV).
In this architecture, a transaction's lifecycle consists of three phases. In the
first phase, \textit{execution}, a client sends the transaction to a set of
nodes, or peers, specified by an endorsement policy.
The transaction is executed by these peers in parallel and its effects in terms
of read and written states are recorded.
Moreover, transactions from different clients may be parallelized during the
execution.
In the second phase, \textit{ordering}, a consensus protocol is used
to produce a totally ordered sequence of endorsed transactions grouped in
blocks.
This order is broadcast to all peers. In the third phase, \textit{validation},
each peer validates the state changes from the endorsed transactions with
respect to the endorsement policy and serializability.


The new EOV architecture limits the execution details of a transaction to the
endorsing peers to enhance confidentiality and exploit concurrency. But such
concurrency comes at the cost of aborting transactions that do not abide
serializability.
\revision{We evaluate the impact of concurrency control in Fabric on the setup described in \Cref{sec:experiment}.
We measure both the raw and effective peak throughputs under both no-op transactions, with no data access, and update transactions, with varying skewness controlled by the zipfian coefficient. 
The raw throughput represents the in-ledger transaction rate, while the effective throughput represents committed transactions by excluding the aborted  transactions from raw throughput. 
In \Cref{fig:intro}, a bar reports the raw throughput, while its blue part reports the effective throughput. 
The raw throughput is constant (677 tps) despite the workload type and request skewness. 
But with higher skewness, a larger proportion of transactions are aborted for serializability.
}
\begin{figure}
  \begin{minipage}{0.99\textwidth}
    \centering     
    \includegraphics[width=0.4\textwidth]{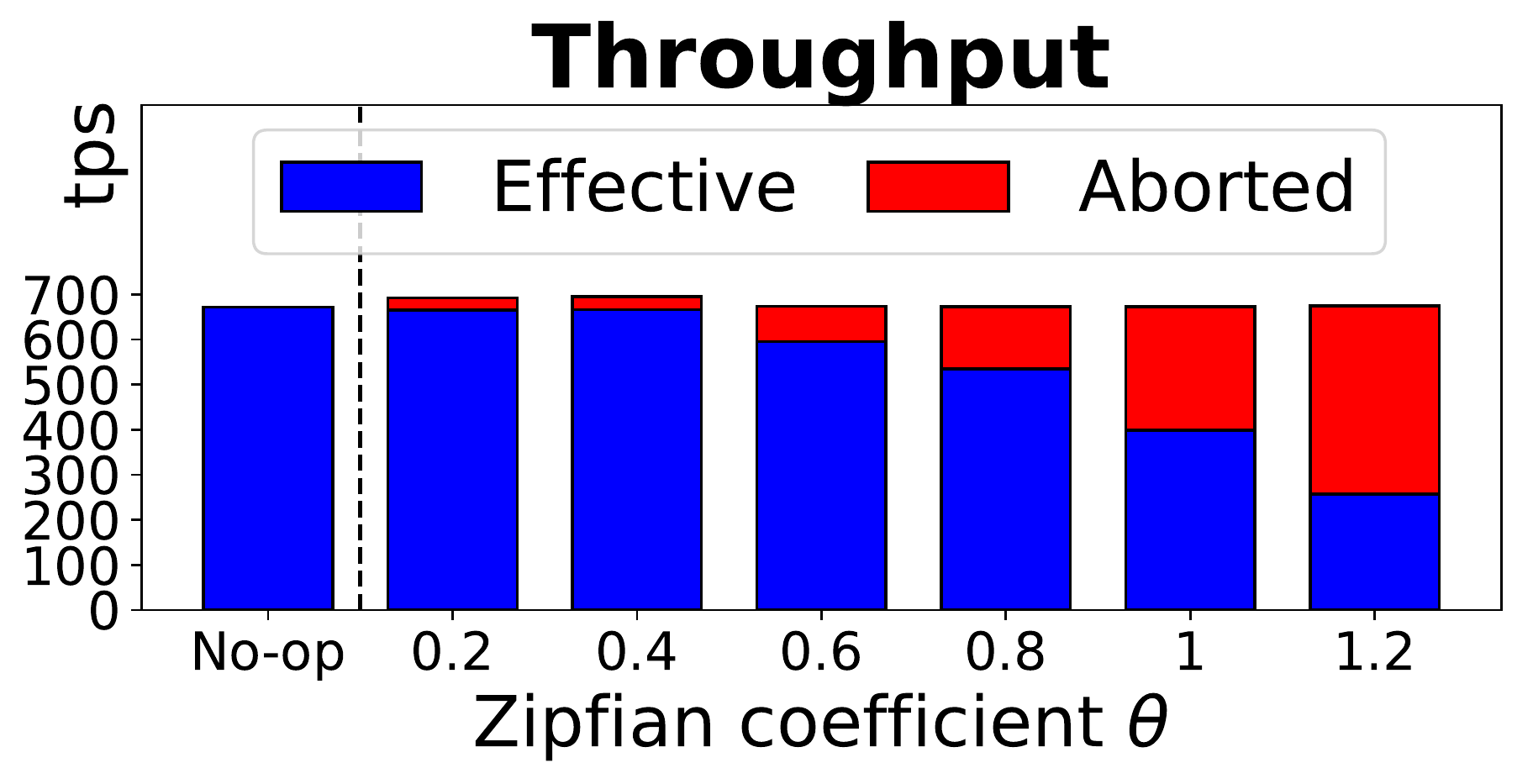}
  \end{minipage}

  \caption{Fabric's raw and effective throughput under both no-op transactions and single modification transactions with varying skewness}
  \label{fig:intro}

\end{figure}

There are two notable directions attempting to solve this issue.
The first is to improve upon Fabric's architecture to enhance its
attainable throughput \cite{nasir2018performance, thakkar2018performance, benhamouda2019supporting,web:fabricSharp}.
\revision{For example, FastFabric proposes to split a node's functionality to alleviate the bottleneck and achieves the highest throughput among all improvements of Fabric~\cite{gorenflo2019fastfabric}. 
}
However, these approaches are implementation-specific and might not generalize well
to other blockchains.
The second direction is to abstract out the transaction lifecycle to reduce
abort rate.
For example, {\fabricPlusplus}~\cite{sharma2019blurring} uses
well-established concurrency techniques from databases to early abort
transactions or reorder them to reconcile the potential conflicts.

Our work corresponds to the second direction, as a major attempt to
\textit{databasify} blockchains.
Here, we take a principled approach to learn from transactional analysis
techniques in databases with optimistic concurrency control (OCC) and apply them
to enhance transaction processing in blockchains.
We formally analyze the behavior of the current implementations of Fabric and {\fabricPlusplus}, and discover that both
achieve \textit{Strong Serializability} \cite{bailis2013highly} (as described in \Cref{sec:serializabilityAnalysis}).
In fact, these implementations are more stringent than
\textit{One-Copy Serializability} (or simply Serializability), as prescribed by the original Fabric protocol~\cite{androulaki2018hyperledger}.
Both systems employ a preventive approach which might over-abort transactions that are still serializable.
In contrast, our proposal consists of a novel reordering technique that
eliminates unnecessary abort due to in-ledger conflicts, with the serializability guarantee established upon our theoretical insights. 
\revision{Our approach does not change Fabric's architecture, therefore it is orthogonal to the aforementioned optimizations, such as FastFabric \cite{gorenflo2019fastfabric} and \cite{nasir2018performance, thakkar2018performance, benhamouda2019supporting,web:fabricSharp}. 
}

In summary, our paper makes the following contributions:

\begin{itemize}[leftmargin=1.8em]
\item We theoretically analyze the resemblance of transaction processing in
  blockchains with EOV architecture and databases with optimistic concurrency
  control (\Cref{sec:resembalance}).
  Based on this resemblance, we analyze the transactional behavior of
  state-of-the-art EOV blockchains, such as Fabric and {\fabricPlusplus}
  (\Cref{sec:serializabilityAnalysis}).
  
\item We propose a novel theorem to identify transactions that can never be
  reordered for serializability 
  (\Cref{sec:reorderabilityAnalysis}).
  Based on this theorem, we propose efficient algorithms to early filter out
  such transactions (\Cref{sec:concurrencyControl}), 
  \revision{with the guarantee of serializability for the remaining transactions after reordering.
  We also discuss the security implications of our proposal (\Cref{sec:securityanalysis}).
  }

\item We implement our proposed algorithms on top of two existing blockchains. First, we use Hyperledger Fabric v1.3 as the base and name our implementation \fabricSharp~(or Fabric\# for short). \revision{Second, we start from FastFabric~\cite{gorenflo2019fastfabric}, which obtains the highest throughput among all optimizations of Fabric, and name our implementation \ffabricSharp~(or FastFabric\# for short).}

\item We extensively evaluate {\fabricSharp} by comparing it with the vanilla
  Fabric, {\fabricPlusplus}, and two other implementations based on database concurrency control techniques from \revision{one standard approach~\cite{CahillRF08}} and a recent proposal by Ding et al~\cite{ding2018improving}.
  The experimental results show that the throughput of {\fabricSharp} is more than \revision{25\%} higher compared to the other systems.
  \revision{In addition, the {\ffabricSharp}'s improvement over FastFabric is up to 66\%.} 
\end{itemize}

The remaining of this paper is structured as follows.
\Cref{sec:background} provides background on EOV blockchains and OCC techniques. 
%
Our theoretical analysis follows in \Cref{sec:theory}, ending with our
reordering algorithm.
\Cref{sec:impl} describes the implementation of our approach.
\Cref{sec:experiment} reports our experimental results.
We review the related work in \Cref{sec:related}, before concluding in
\Cref{sec:conclusion}.



\begin{figure*}[h!]
  \begin{subfigure}{0.71\textwidth}
    \centering
    \includegraphics[width=0.99\textwidth]{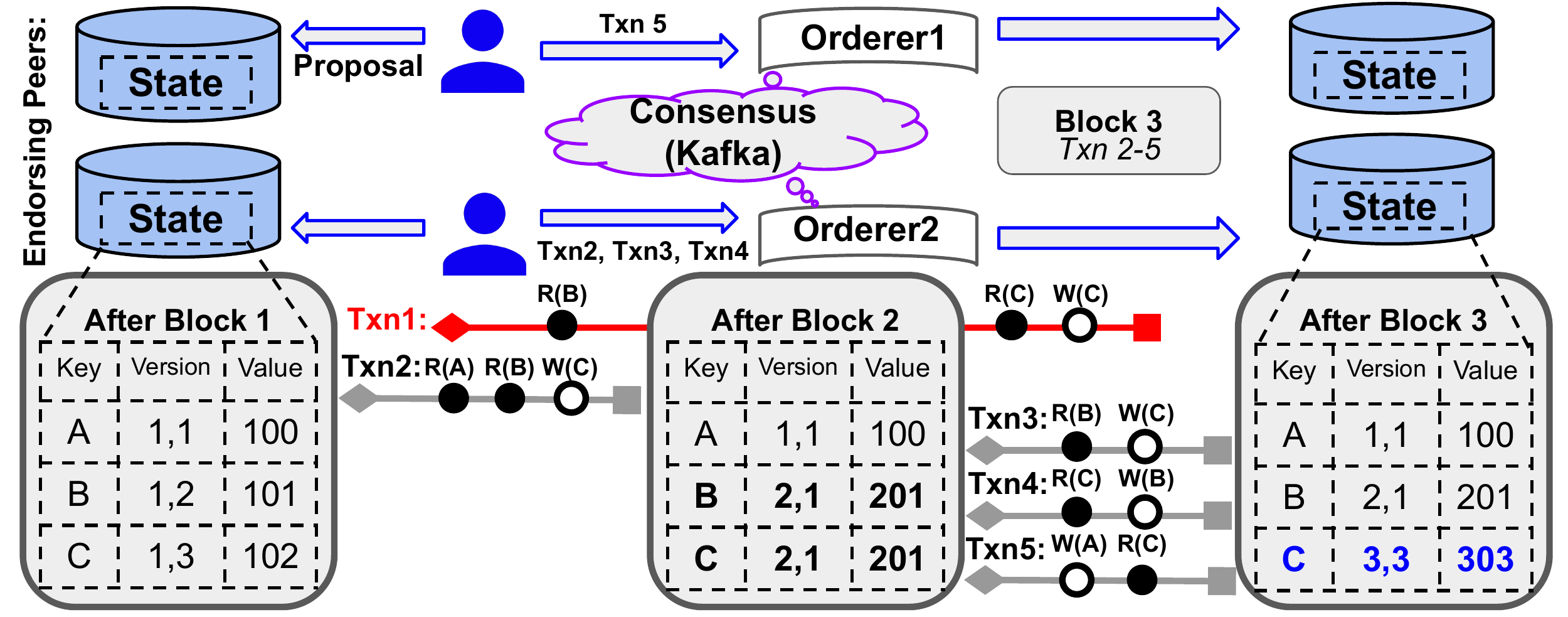}
    \caption{}
    \label{fig:background_fabric}
  \end{subfigure}
  \begin{subfigure}{0.28\textwidth}
    \centering
    \includegraphics[width=0.75\textwidth]{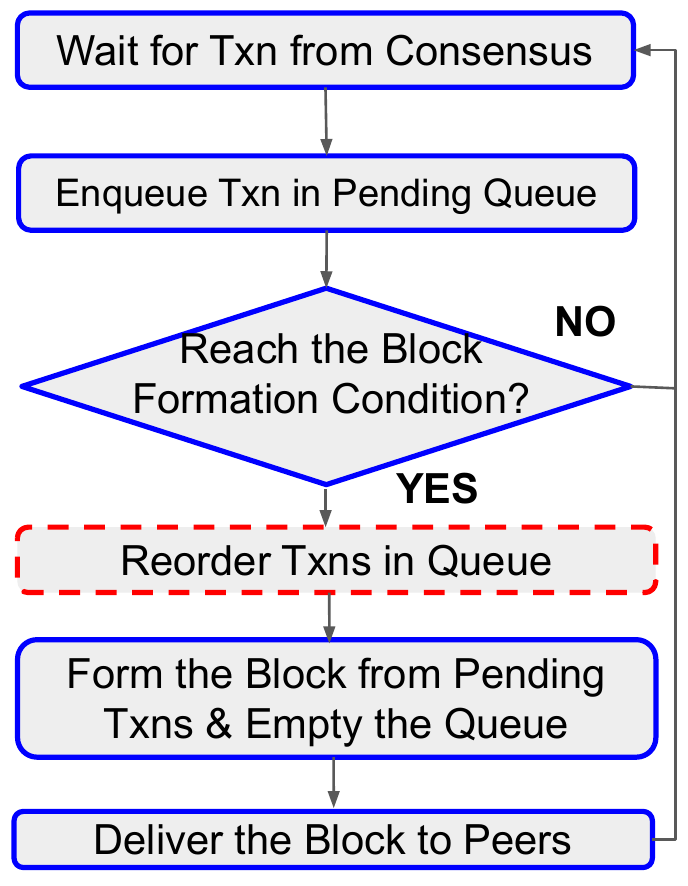}
    \caption{}
    \label{fig:background_orderer}
  \end{subfigure}

  \caption{\revision{(a) Example of transaction workflow in Fabric. An arrow represents the lifespan of a transaction's execution (simulation), e.g, Txn1 starts its execution immediately after block 1 and finishs its simulation after block 2. (b) Procedures replicated in each Fabric Orderer.
  {\fabricPlusplus} introduces a reordering step before the block formation to reduce the transaction abort rate.}}
  \label{fig:background}

\end{figure*}

\section{Background}
\label{sec:background}

\subsection{\revision{EOV~architecture~in~Fabric~and~{\fabricPlusplus}}}
\label{sec:background_fabric}

Hyperledger Fabric~\cite{androulaki2018hyperledger} is a state-of-the-art
permissioned blockchain that features a modular design based on the EOV architecture.
%
\revision{{\fabricPlusplus}~\cite{sharma2019blurring} is an optimization of Fabric, which reorders transactions after consensus to reduce the abort rate.
}
A Fabric/\revision{{\fabricPlusplus}} blockchain is run by a set of authenticated nodes, whose identity is
provided by a membership service.
A node in this blockchain has one of the following three roles:
(i) \textit{client} which submits a transaction proposal for execution,
(ii) \textit{peer} which \textit{executes} and \textit{validates}
transaction proposals,
or (iii) \textit{orderer} which \textit{orders} transactions and batches
them in blocks.
Transaction order is determined collectively by all orderers in
the blockchain based on a consensus protocol.
%

The state of a blockchain after forming a block is maintained by a
versioned key-value store.
Each entry in this store is a tuple (\texttt{key}, \texttt{ver},
\texttt{val}), where \texttt{key} is a unique name representing the entry, and
\texttt{ver} and \texttt{val} are the entry's latest version and value, respectively.
Moreover, \texttt{ver} is a pair consisting of the sequence number of the block
and the transaction that updated the entry.
%
%
For example, in \Cref{fig:background_fabric}, the entry \entry{C}{2,1}{201} in
the state after block $2$ indicates that the key \texttt{C} contains the latest
value $201$ which was lastly updated by the 1st transaction in block $2$.

\begin{table}[tp]
  \centering
  \caption{Summary of transactions in \Cref{fig:background}.
    Staled reads and installed writes are
    respectively marked in {\color{red} red} and {\color{blue} blue} colors. The
    symbols {\vmark}, {\xmark}, \textbf{N.A.} respectively indicate committed,
    aborted, or not-allowed transactions.}

  \label{tab:simulation}
  \vspace{4pt}
  \small
  \setlength\tabcolsep{2.3pt}
  \begin{tabular}{|c|c|cc|cc|c|l|c|l|c|l|}
  \hline
  \multicolumn{2}{|c|}{}
  & \multicolumn{2}{c|}{\textbf{Txn1}}
  & \multicolumn{2}{c|}{\textbf{Txn2}}
  & \multicolumn{2}{c|}{\textbf{Txn3}}
  & \multicolumn{2}{c|}{\textbf{Txn4}}
  & \multicolumn{2}{c|}{\textbf{Txn5}}
  \\

  \hline

  \multirow{2}{*}{\textbf{Readset}} & Key
  & B & C
  & A & {\color{red} \textbf{B}}
  & \multicolumn{2}{c|}{B}
  & \multicolumn{2}{c|}{{\color{red} \textbf{C}}}
  & \multicolumn{2}{c|}{{\color{red} \textbf{C}}}
  \\

  & Version
  & 1,2 & 2,1
  & 1,1 & {\color{red} \textbf{1,2}}
  & \multicolumn{2}{c|}{2,1}
  & \multicolumn{2}{c|}{{\color{red} \textbf{2,1}}}
  & \multicolumn{2}{c|}{{\color{red} \textbf{2,1}}}
  \\

  \hline

  \multirow{2}{*}{\textbf{Writeset}} & Key
  & \multicolumn{2}{c|}{C}
  & \multicolumn{2}{c|}{C}
  & \multicolumn{2}{c|}{{\color{blue} \textbf{C}}}
  & \multicolumn{2}{c|}{B}
  & \multicolumn{2}{c|}{A}
  \\

  & Value
  & \multicolumn{2}{c|}{301}
  & \multicolumn{2}{c|}{302}
  & \multicolumn{2}{c|}{{\color{blue} \textbf{303}}}
  & \multicolumn{2}{c|}{304}
  & \multicolumn{2}{c|}{305}
  \\

  \hline

  \multirow{2}{*}{\textbf{Commit status}} & Fabric
  & \multicolumn{2}{c|}{\textbf{N.A.}}
  & \multicolumn{2}{c|}{\xmark}
  & \multicolumn{2}{c|}{\vmark}
  & \multicolumn{2}{c|}{\xmark}
  & \multicolumn{2}{c|}{\xmark}
  \\

  & {\fabricPlusplus}
  & \multicolumn{2}{c|}{{\xmark}}
  & \multicolumn{2}{c|}{{\xmark}}
  & \multicolumn{2}{c|}{{\xmark}}
  & \multicolumn{2}{c|}{\vmark}
  & \multicolumn{2}{c|}{\vmark}
  \\

  \hline
\end{tabular}
\vspace{-8pt}
\end{table}

In Fabric/\revision{{\fabricPlusplus}}, the workflow of a transaction consists of three phases:
execution, ordering, and validation.
We elaborate on these phases below, using the example in
\Cref{fig:background_fabric}.


\textbf{Execution}. In this phase, clients propose transactions
  consisting of smart contract invocations to a set of endorsing peers,
  which are selected by an endorsement policy.
  Each endorsing peer executes transaction proposals concurrently and
  speculatively and returns the simulation results together with its endorsement
  signature.
  The results contain two value sets called the \textit{readset} and the
  \textit{writeset} which respectively represent the version dependencies (all
  keys read along with their version numbers) and the state updates (all keys
  modified along with their new values) produced by the simulation.
  For example, the readset and writeset of transactions in
  \Cref{fig:background_fabric} are summarized in \Cref{tab:simulation}.
  Throughout the execution, a transaction holds a read lock on the
  state database to guarantee that read values are the latest.
  Transactions that read across blocks, such as \texttt{Txn1} in
  \Cref{fig:background_fabric}, are not allowed in Fabric.
  \revision{In contrast, {\fabricPlusplus} optimistically removes this lock for more parallelism but aborts transactions that read across blocks.
  }
  After a client collects enough identical simulation results as required
  by the endorsement policy, it packages them into a single transaction and
  submits it to orderers.

\textbf{Ordering}. In this second phase, orderers receive
  transactions and sequence them into a total order to form a block, 
  \revision{
  as shown in Figure~\ref{fig:background_orderer}.
  Each orderer may belong to different administrative domains and receive different transaction proposals from various clients. 
  But all orderers rely on a single consensus protocol to establish a common transaction order.
  Fabric/{\fabricPlusplus} outsources this consensus service to Kafka. 
  With the consistent transaction stream from the consensus, each orderer employs the same block formation protocol to batch transactions into blocks, and consequently delivers them to peers.  
  A block is formed when the number of pending transactions reach the threshold or a timeout triggers. 
  For example, in \Cref{fig:background_fabric}, Orderer1 receives \texttt{Txn5} and Orderer2 receives \texttt{Txn2},  \texttt{Txn3}, and \texttt{Txn4}. 
  They send the transactions to the consensus service and receive the same transaction order. 
  Based on this order, both orderers package the transactions into identical blocks, i.e., block 3 with  \texttt{Txn2} to  \texttt{Txn5}, given that the protocol limits the maximum number of transactions per block to 4. 
  }

\textbf{Validation}. This phase is executed by each peer after
  a block has been retrieved from orderers.
  Transactions in a block are sequentially validated based on the
  corresponding endorsement policy and transaction serializability.
  The serializability of a transaction is tested by inspecting the
  staleness of its readset. The transaction is marked as invalid if it
  reads a key whose version at the read time is inconsistent (or older) than the
  latest version.
  For example, in \Cref{fig:background_fabric}, transaction \texttt{Txn2} in
  block $2$ is unserializable since it reads key \texttt{B} with version
  $(1,2)$ from block $1$, which is inconsistent with the latest version $(2,1)$
  in block $2$.
  Suppose that \texttt{Txn3} passes the serializability test and updates the version of key \texttt{C} from $(2,1)$ to $(3,3)$ in block $3$. Then, transactions \texttt{Txn4} and \texttt{Txn5} become invalid, since they both read an inconsistent version of key \texttt{C} in block $2$.
  Hence, after this validation phase, only transaction \texttt{Txn3} in
  block $3$ is committed, while transactions \texttt{Txn2}, \texttt{Txn4}, \texttt{Txn5} are aborted.
  \revision{To satisfy the serializability constraint, {\fabricPlusplus} introduces a reordering step immediately before block formation but after consensus.
  The reordering is based on the commit order determined by the consensus and the accessed records in the transactions. 
  For example, each orderer in {\fabricPlusplus} puts \texttt{Txn3} behind \texttt{Txn4} and  \texttt{Txn5}. Then, \texttt{Txn4} and \texttt{Txn5} are committed while \texttt{Txn3} is aborted. 
  Hence, {\fabricPlusplus} commits one more transaction than Fabric. 
  }

\subsection{Optimistic Concurrency Control in Databases}


Unlike pessimistic concurrency control, the OCC technique does not hold locks to
regulate transactional interference.
Instead, each transaction has a unique \textit{start timestamp} assigned to it
from a global atomic clock.
All queries reflect the state snapshot of the database at the start timestamp,
without observing later changes.
Each transaction is also assigned an \revision{\textit{end timestamp}}. 
Before committing, the database system checks the validity of a transaction
based on these two timestamps and the accessed records.
OCC can easily achieve \textit{Snapshot Isolation}, which disallows concurrent
transactions updating the same key \cite{berenson1995critique}.
Considering the fact that Snapshot Isolation suffers anomalies such as
Lost Update and Write Skew, a number of attempts have been
made to transform Snapshot Isolation to Serializable level
\cite{fekete2005making, yabandeh2012critique, bornea2011one}.
  


\section{Theoretical Analysis}
\label{sec:theory}

\begin{figure*}[tp] \centering
  \begin{subfigure}{0.7\textwidth}
    \includegraphics[width=0.99\textwidth]{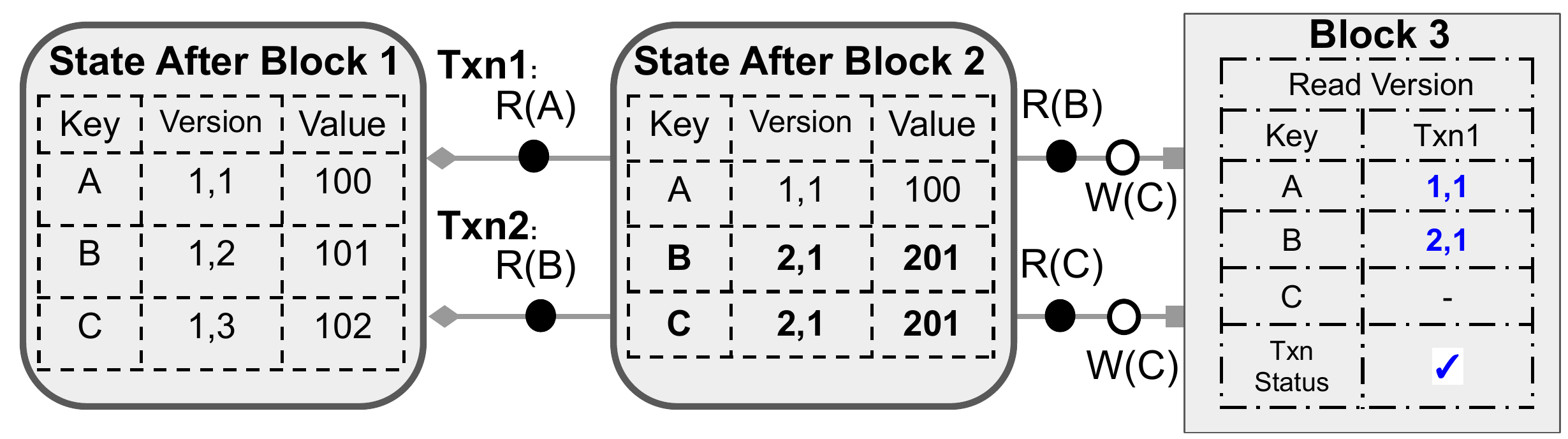}
    \caption{Txn1, which reads across blocks, is snapshot consistent and can be
    scheduled with serializability. Txn2 is not as its early-read key B
    is updated before its execution ends.}
    \label{fig:theory_snapshot}
  \end{subfigure}\hfill
  \begin{subfigure}{0.27\textwidth}
    \includegraphics[width=0.99\textwidth]{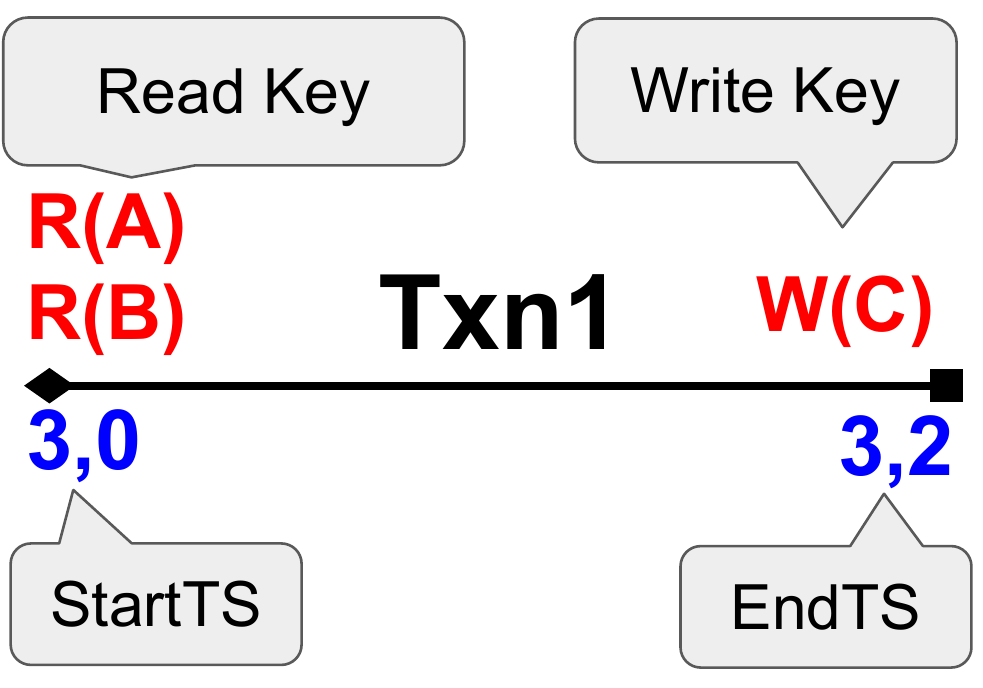}
    \caption{Concise notation to represent a transaction}
    \label{fig:theory_notation}
  \end{subfigure}

  \caption{Example of transactions reading across blocks}

\end{figure*}

In this section, we first describe the resemblance of transaction processing
techniques in EOV blockchains and OCC databases.
Then, we use the transactional analysis method of OCC databases to
reason about the serializability behavior of EOV blockchains, such as Fabric and {\fabricPlusplus}.
Finally, we propose a reordering-based concurrency control algorithm for ordering
serializable transactions in EOV blockchains, 
\revision{along with the discussion on its security implications.
}

\subsection{Resemblance~in~Transaction~Processing}
\label{sec:resembalance}
%
Similar to database systems where the concept \textit{database snapshot} is used
to describe a read-only, static view of a database~\cite{kung1981optimistic}, in blockchains, we can define the similar concept of \textit{blockchain snapshot} as follows.

\begin{definition}[Blockchain snapshot]
  \label{defn:snapshot}
  A blockchain snapshot is the state of a blockchain after a block is committed.
  Let $M$ be the sequence number of the committed block, then the
  corresponding snapshot is denoted as $M$ and is said to have the sequence
  number $(M{+}1,0)$~\footnote{We use the two-value tuple with 0 fixed for the
    second element. This is to facilitate the ordering relations $<$ of sequence numbers of blockchain snapshots and
    transaction timestamps.}.
\end{definition}



\begin{definition}[Snapshot consistency]
  \label{defn:inconsistentExe}
  A transaction is snapshot consistent if there exists a blockchain snapshot M
  from which all the transaction's records are read.
\end{definition}



Transactions in Fabric satisfy snapshot consistency since Fabric uses a lock to
ensure the simulation is done against the latest state.
\revision{{\fabricPlusplus} optimistically removes the lock but early aborts transactions which read across blocks.
Hence, it also satisfies the snapshot consistency. 
}
However, eliminating transactions based on cross-block reading might
lead to over-aborting snapshot consistent transactions.

  For examle, in \Cref{fig:theory_snapshot}, \texttt{Txn1} reads key
  \texttt{A} of version $(1,1)$ in snapshot $1$ and key \texttt{B} of
  version $(2,1)$ in snapshot $2$.
  These versions are the same as the versions of keys \texttt{A} and \texttt{B} in snapshot $2$.
  Hence, \texttt{Txn1} is \textit{snapshot consistent} with block snapshot $2$.
  In contrast, transaction \texttt{Txn2}, which also reads across blocks, does
  not achieve snapshot consistency because the value of previously read key B
  changes in block 2.
  %

%

\begin{proposition}
  \label{proposition:crossblockRead}
  There exist snapshot-consistent transactions that read across blocks. For such
  a transaction, its block snapshot is determined by its last read operation.
\end{proposition}

\begin{proof}
  \texttt{Txn1} in \Cref{fig:theory_snapshot} is a witness example.
  We have shown in the previous example that \texttt{Txn1} reads across
  blocks 1 and 2, and it is still consistent with block snapshot 2.
\end{proof}

\Cref{proposition:crossblockRead} shows that a legitimate transaction in an
EOV blockchain can read across blocks, if their states are consistent.
This makes the EOV blockchain similar to an OCC database, as the
latter also reads from consistent states determined by the transaction's start
timestamp.
We also observe that the blockchain's sequence numbers have similar properties with
databases' timestamps, such as atomicity, monotony, total order, and
unique mapping to snapshots.
Therefore, we define the timestamps of blockchain transactions using their
sequence numbers.

\begin{definition}[Start timestamp]
  \label{defn:start-timestamp}
  The start timestamp of transaction \texttt{Txn}, denoted by \startTS{Txn}, is
  the sequence number of its read snapshot.
\end{definition}

\revision{
\begin{definition}[End timestamp]
  \label{defn:commit-timestamp}
  The end timestamp of transaction \texttt{Txn}, denoted by \commitTS{Txn}, is its sequence number in the block, determined by the consensus. 
  %
\end{definition}
}

For example, in \Cref{fig:theory_snapshot}, \texttt{Txn1} has $\startTS{Txn1} =
(3, 0)$ and $\commitTS{Txn1} = (3,1)$, since it lastly reads from block $2$ and
occupies the first position in block $3$.
For brevity, in later paragraphs, we use the notation presented in
\Cref{fig:theory_notation} to denote a transaction.
\revision{Moreoever, the sequence numbers of transactions' start or end timestamps are lexicographically ordered, e.g., $(2,1) < (2,2) = (2,2) < (3,0)$.
}



\begin{definition}[Concurrent transactions]
  \label{defn:concurrent-transaction}
  \revision{Two transactions \texttt{Txn1} and \texttt{Txn2} are said to be concurrent if their executions overlap. 
  To be specific, if \texttt{Txn1} ends earlier than \texttt{Txn2} (i.e., $\commitTS{Txn1} < \commitTS{Txn2}$), then \texttt{Txn2} must start before \texttt{Txn1} ends (i.e., $\startTS{Txn2} < \commitTS{Txn1}$). Otherwise, if \texttt{Txn2} ends earlier than \texttt{Txn1} (i.e., $\commitTS{Txn2} < \commitTS{Txn1}$), then \texttt{Txn1} must start before \texttt{Txn2} ends (i.e., $\startTS{Txn1} < \commitTS{Txn2}$).
  }
\end{definition}


\begin{proposition}
  \label{proposition:concurrency}  
  Each pair of transactions in the same block are concurrent. 
\end{proposition} 

\begin{proof}
  Suppose two transactions \texttt{Txn1} and \texttt{Txn2} are committed in the
  same block $M$ at position $p$ and $q$, respectively, where $p < q$.
  Since the latest block that \texttt{Txn2} can read from is $M{-}1$, we have
  that: $\startTS{Txn2} \le (M,0) < \commitTS{Txn1} = (M,p) < \commitTS{Txn2} =
  (M,q)$.
  Hence, \texttt{Txn1} and \texttt{Txn2} are concurrent.
\end{proof}

\begin{proposition}
  \label{proposition:nonconcurrency}  
  The reverse of \Cref{proposition:concurrency} is not true: there are concurrent
  transactions not belonging to the same block.
\end{proposition} 

\begin{proof}
  We present a witness example in \Cref{fig:theory_concurrency}, where
  transactions \texttt{Txn1} and \texttt{Txn2} respectively belong to block $M$
  and $M{+}1$. However, \texttt{Txn2} reads from a block earlier than $M$.
  Hence, we have: $\startTS{Txn2} \le (M,0) < \commitTS{Txn1} = (M,1) <
  \commitTS{Txn2} = (M{+}1,1)$.
  Therefore, \texttt{Txn1} and \texttt{Txn2} are concurrent.
\end{proof}

\begin{figure}[tp] \centering
  \includegraphics[width=0.6\textwidth]{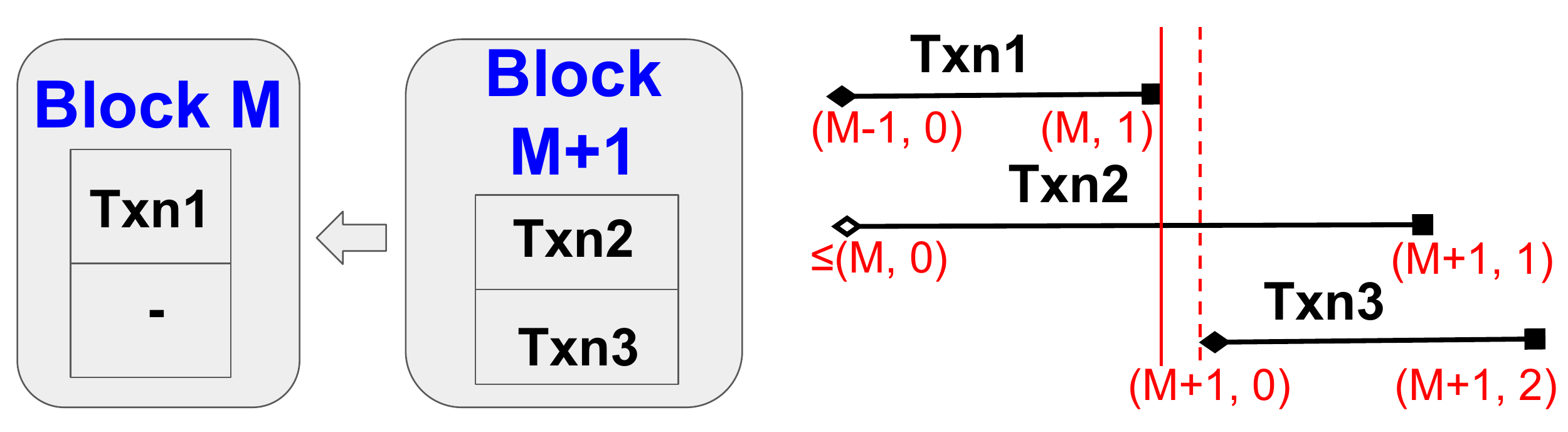}
  \caption{\texttt{Txn2} and \texttt{Txn3} are in the same block and
    concurrent.
    \texttt{Txn1} and \texttt{Txn2} are in different blocks, but they
    are still concurrent.
    \texttt{Txn1} and \texttt{Txn3} are not concurrent.}
  \label{fig:theory_concurrency}
\end{figure}

\revision{From the above two propositions, concurrency does not only occur between transactions within the same block.  
{\fabricPlusplus} fails to consider dependencies among transactions across blocks.
Hence, its reordering effect is limited. 
}

\subsection{Serializability Analysis}
\label{sec:serializabilityAnalysis}

\begin{figure}[tp]
  \begin{subfigure}{0.33\textwidth}
      \includegraphics[width=0.99\textwidth]{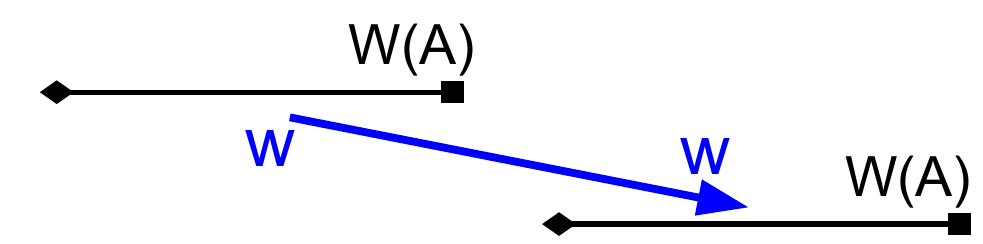} 
      \caption{\textit{n-ww}}
  \end{subfigure}
  \begin{subfigure}{0.33\textwidth}
      \includegraphics[width=0.99\textwidth]{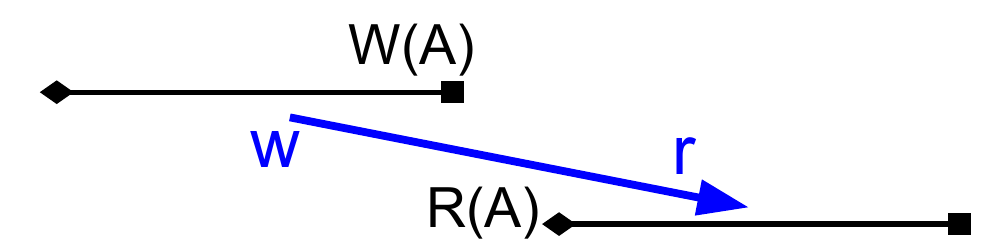}
      \caption{\textit{n-wr}}
  \end{subfigure}
  \begin{subfigure}{0.33\textwidth}
      \includegraphics[width=0.99\textwidth]{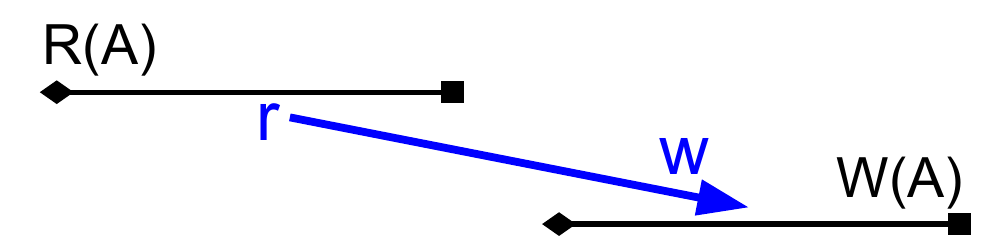} 
      \caption{\textit{n-rw}}
  \end{subfigure}
  \begin{subfigure}{0.33\textwidth}
      \includegraphics[width=0.99\textwidth]{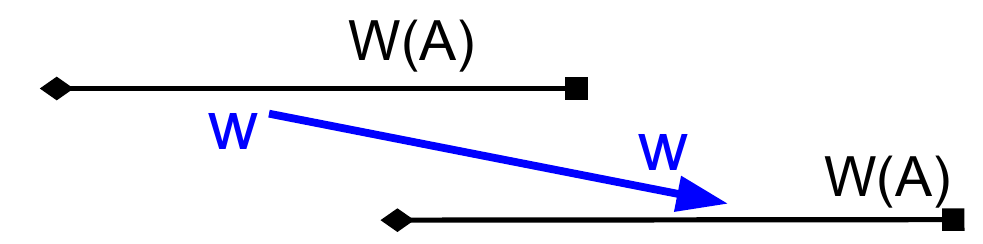} 
      \caption{\textit{c-ww}}
  \end{subfigure}
  \begin{subfigure}{0.33\textwidth}
      \includegraphics[width=0.99\textwidth]{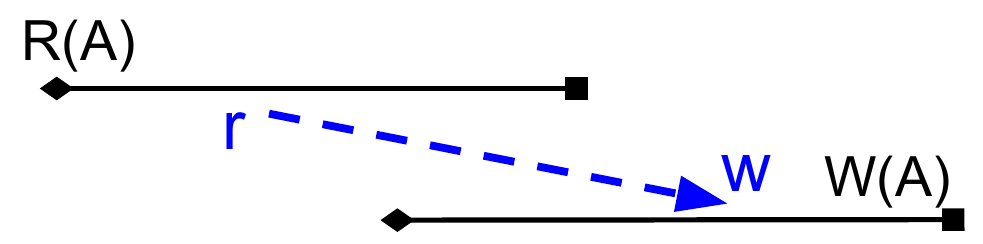} 
      \caption{\textit{c-rw}}
  \end{subfigure}
  \begin{subfigure}{0.33\textwidth}
      \includegraphics[width=0.99\textwidth]{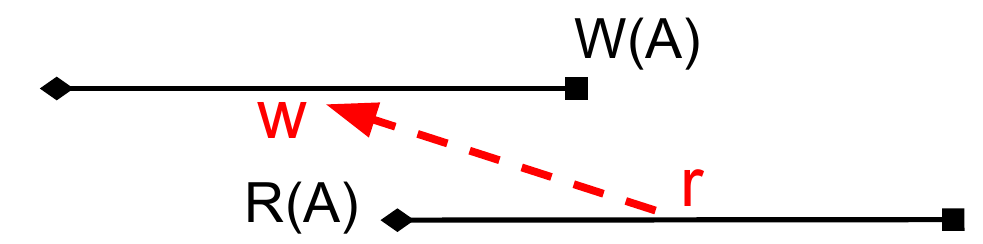} 
      \caption{\textit{anti-rw}}
  \end{subfigure}  
  \caption{Six canonical dependencies between snapshot transactions.
    Here, (a), (b), and (c) are non-concurrent and (d), (e), and (f)
    are concurrent.}
  \label{fig:theory_dependencies}
\end{figure}

\Cref{fig:theory_dependencies} shows all six scenarios of canonical transaction
dependency (or conflict) between snapshot transactions, as described by
\cite{fekete2005making}.
Among them, three dependencies, namely \scenario{n-ww}, \scenario{n-wr}, and
\scenario{n-rw} are between non-concurrent transactions.
The other three dependencies, namely \scenario{c-ww}, \scenario{c-rw}, and \scenario{anti-rw} are between concurrent transactions.
According to the conflict serializability theorem in
\cite{weikum2001transactional}, the effect of a serializable transaction
schedule is equivalent to any serialized transaction history that respects
dependency order.
Note that the dependency graph of the serializable transaction schedule must be acyclic.
%




\begin{definition}[Strong Serializability]
  \revision{A schedule of transactions is \textit{Strong Serializable} if its effect is
  equivalent to the serialized history, which conforms to the transactions'
  commit order determined by their end timestamps.
  }
\end{definition}

\begin{theorem}   
  \label{theory:strict}
  \revision{A schedule of transactions without \scenario{anti-rw} achieves Strong Serializability.
  }
\end{theorem}

\begin{proof}
  \revision{We first prove that any transaction schedule without \scenario{anti-rw} achieves Serializability.
  By contradiction, suppose that such a transaction schedule does not achieve Serializability.
  Then, in the schedule there must be a subset of transactions with a dependency cycle,
  in which the last committed transaction is denoted by \texttt{Txn}.
  Then \texttt{Txn} must exhibit an \scenario{anti-rw} dependency because
  \scenario{anti-rw} is the only one among all six dependencies that relates
  later transactions to earlier ones.
  But this contradicts our assumption. 
  Hence, the transaction schedule is serializable.
  Next, we prove that it also achieves Strong Serializability.
  Since the order of the five remaining dependencies is consistent to their
  commit order, the serialized history that respects the commit order also
  respects the dependency order.
  According to the conflict serializability theorem in
  \cite{weikum2001transactional}, this serialized transaction history has the
  equivalent effect of the serializable schedule.
  Hence, the transaction schedule is Strong Serializable.
  }
\end{proof} 

\revision{We remark that Fabric/{\fabricPlusplus} do not allow \scenario{anti-rw} between two transactions because the latter transaction would read an old version of the updated key, hence, it must be aborted.
Based on Theorem~\ref{theory:strict}, transactions in Fabric/{\fabricPlusplus} satisfy Strong Serializability, which is more stringent than Serializability~\cite{androulaki2018hyperledger}. This opens up the opportunity to reduce the transaction abort rate. 
}

\subsection{Reorderability Analysis}
\label{sec:reorderabilityAnalysis}
\revision{Under Serializability instead of Strong Serializability,}
we formally analyze the reorderability of transactions in EOV blockchains.
%
We focus on determining a serializable schedule by switching the commit order of
pending transactions.

\begin{lemma} 
  \label{lemma:reorder_concurrency}
  In blockchains, reordering can only happen between concurrent transactions. 
\end{lemma}

\begin{proof} 
  Assume transaction reordering occurs between two non-concurrent transactions.
  These transactions are committed in different blocks, due to the
  contra-positive of \Cref{proposition:concurrency}. Switching their order means
  changing a previously committed block, which is impossible in blockchains due
  to their immutability.
\end{proof}

\begin{lemma} 
  \label{lemma:reorder_impact}
  A transaction does not change its concurrency relationship with respect to
  others after reordering.
\end{lemma}

\begin{proof}
  Assume the next block's sequence number is $M$.
  For any pending transaction \texttt{Txn}, we have: $\startTS{Txn} \le (M,0) <
  \commitTS{Txn}$.
  Other transactions are classified into three cases.
  (i) For any non-concurrent transaction \texttt{Txn1}, we have:
  $\commitTS{Txn1} < \startTS{Txn}$.
  Since reordering does not affect $\startTS{Txn}$, the non-concurrency between
  \texttt{Txn} and \texttt{Txn1} still holds.
  (ii) For any concurrent transaction \texttt{Txn2} committed earlier than block
  $M$, we have: $\startTS{Txn} < \commitTS{Txn2} < (M,0) < \commitTS{Txn}$.
  Since reordering cannot move the commit time of \texttt{Txn} before $(M,0)$,
  \texttt{Txn2} and \texttt{Txn} remain concurrent.
  (iii) For any pending transaction \texttt{Txn3},  we have either $\startTS{Txn} < (M, 0) < \commitTS{Txn3} < \commitTS{Txn}$, or $\startTS{Txn3} < (M, 0) < \commitTS{Txn} < \commitTS{Txn3}$.
  Hence, \texttt{Txn} and \texttt{Txn3} remain concurrent after reordering.
\end{proof}

The above \Cref{lemma:reorder_concurrency} ensures that reordering does not
impact non-concurrent transactions and their dependencies.
\Cref{lemma:reorder_impact} ensures that non-concurrent transactions are not
introduced by reordering.
Therefore, we restrict our analysis to concurrent dependencies.
We describe the dependency order of concurrent transactions using the two lemmas below.

\begin{figure}[tp]
      \centering
      \begin{subfigure}{0.45\textwidth}
        \includegraphics[width=0.8\textwidth]{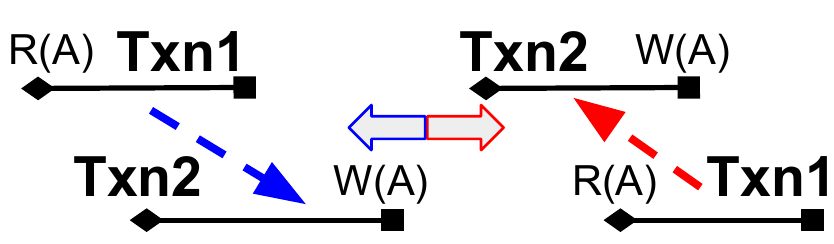}

      \end{subfigure}      
      \begin{subfigure}{0.45\textwidth}
        \includegraphics[width=0.8\textwidth]{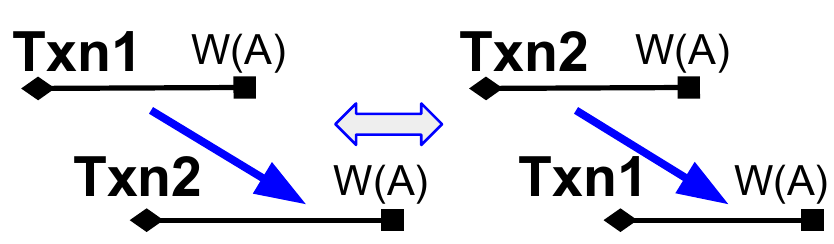}
      \end{subfigure}

      \caption{Dependency order preserves between \scenario{c-rw},
        \scenario{anti-rw} but not \scenario{c-ww} when switching commit order}
      \label{fig:theory_reorder}
\end{figure}

\begin{lemma} 
  \label{lemma:reorder_rw}
  If two transactions \texttt{Txn1} and \texttt{Txn2} exhibit \scenario{c-rw} or
  \scenario{anti-rw} dependency, switching their commit order does not affect
  their dependency order.
\end{lemma}

\begin{proof}
  When \texttt{Txn1} and \texttt{Txn2} exhibit \scenario{c-rw} (or
  \scenario{anti-rw}) dependency, if we switch their commit order, they will
  exhibit \scenario{anti-rw} (or \scenario{c-rw}) dependency, as illustrated in the left side of \Cref{fig:theory_reorder}.
  Consequently, in both two cases, their dependency order remains the same,
  i.e., \texttt{Txn1} reads a key which will be written later by \texttt{Txn2}.
\end{proof}

\begin{lemma} 
  \label{lemma:reorder_ww}
  If two transactions \texttt{Txn1} and \texttt{Txn2} exhibit \scenario{c-ww}
  dependency, switching their commit order flips their dependency order.
\end{lemma}

\begin{proof}
  When \texttt{Txn1} and \texttt{Txn2} exhibit \scenario{c-ww} dependency,
  \texttt{Txn1} writes to a key which will be over-written by \texttt{Txn2}.
  If their commit order is switched, then \texttt{Txn2} and \texttt{Txn1} will
  exhibit \scenario{c-ww} dependency, as illustrated in the right side of
  \Cref{fig:theory_reorder}.
  Now, \texttt{Txn2} writes to a key which will be over-written by \texttt{Txn1}.
  Consequently, the dependency order of \texttt{Txn1} and \texttt{Txn2} is flipped.
\end{proof}

Finally, we present a theorem on reordering transactions containing a dependency
cycle. This theorem is utilized in \Cref{sec:concurrencyControl} to
design our novel fine-grained concurrency control algorithm.

\begin{theorem}
  \label{theory:unreorderable}
  \revision{A transaction schedule cannot be reordered to be
  serializable if there exists a cycle with no \scenario{c-ww} dependencies involving pending transactions.
  }
\end{theorem}

\begin{proof}
\revision{
  We classify the dependencies in the cycle into two categories. 
  The first category includes those involving at least one committed transaction in the dependecy. 
  Due to the immutability of blockchains, reordering does not impact these dependencies, because the relative commit order of two transactions is fixed.
  The second category includes all dependencies between a pair of pending transactions. 
  For each dependency, its corresponding transactions must be concurrent, otherwise, the preceding transaction would be committed. 
  Due to the fact that the pending transactions are concurrent and the absence of \scenario{c-ww}, the order switching can only happen between conflicting transactions with \scenario{c-rw} or \scenario{anti-rw}.
  Their dependency order preserves despite being reordered (\Cref{lemma:reorder_rw}).
  Hence, the cyclic schedule remains unserializable, as shown in
  \Cref{fig:theory_cycle_rw}.
  }
\end{proof}

\revision{However, a transaction schedule can be reordered to be serializable if there exists a cycle with one \scenario{c-ww} conflict between pending transactions. Due to \Cref{lemma:reorder_ww}, their dependency order can be flipped. We present this scenario in \Cref{fig:theory_cycle_ww}, where a cyclic schedule formed by \texttt{Txn1}, \texttt{Txn2} and \texttt{Txn3} becomes serializable by switching the commit order of \texttt{Txn2} and \texttt{Txn3}, which exhibit \scenario{c-ww} dependency.}

\begin{figure}[tp]
	\centering
    \begin{subfigure}{0.45\textwidth}
    	\centering
      \includegraphics[width=0.85\textwidth]{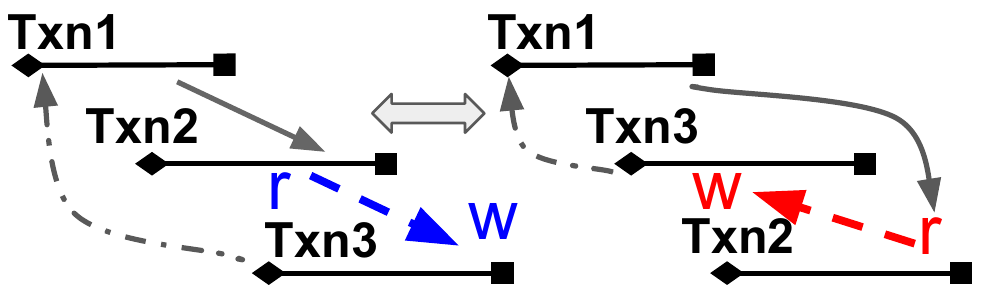}

      \caption{An unreorderable schedule without \scenario{c-ww}}
      \label{fig:theory_cycle_rw}

    \end{subfigure}
    \begin{subfigure}{0.45\textwidth}
    	\centering
      \includegraphics[width=0.85\textwidth]{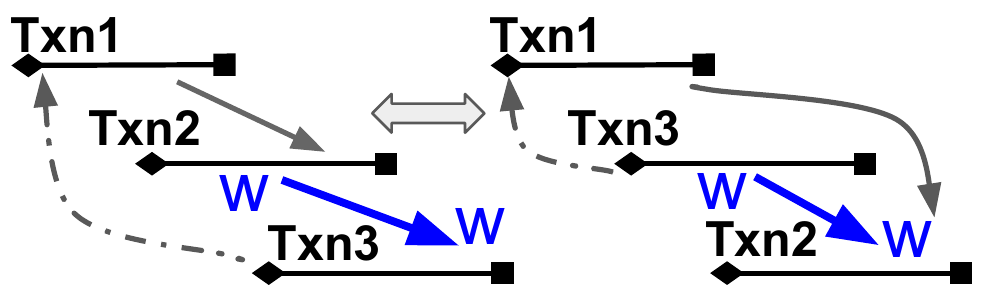}

      \caption{A reorderable schedule with \scenario{c-ww}}
      \label{fig:theory_cycle_ww}
    \end{subfigure}

    \caption{Transaction schedule reorderability}

\end{figure}

\subsection{Fine-grained Concurrency Control}
\label{sec:concurrencyControl}

\Cref{theory:unreorderable} states that a cyclic transaction schedule without
\scenario{c-ww} among pending transactions can never be serializable despite reordering.
Based on this insight, we formulate the following three steps for
fine-grained concurrency control in EOV blockchains.

\begin{itemize}[leftmargin=1.8em]

\item For a new transaction, we first consider all dependencies, except
  \scenario{c-ww}, among all pending transactions (including the new
  transaction).
  Then, we directly drop the new transaction if there is a dependency cycle.

\item On block formation, we retrieve the pending transaction order that
  respects all the computed dependencies.

\item Finally, we restore \scenario{c-ww} dependencies on pending transactions
  based on the retrieved schedule.

\end{itemize}

Note that \scenario{c-ww} dependency restoration is still necessary, 
as future unserializable transactions may encounter a cycle with a \scenario{c-ww}
dependency which involves committed transactions. 
But their commit as well as the dependency order is already fixed.
However, the dependency graph still remains acyclic after the restoration. 

We outline our fine-grained concurrency control in Algorithms
\ref{alg:simulation}, \ref{alg:txn}, and \ref{alg:blk}, while the implementation
details are presented in \Cref{sec:impl}.
We use the notation $A \unioneq B$ to represent the self-assignment with union $A := A \cup B$.
Here, we argue that the topological sort in \Cref{alg:blk} always has a solution
since the transaction dependency graph $G$ is guaranteed to be acyclic by \Cref{alg:txn}.
Even the sub-graph containing only the pending transactions, $P$, is a directed
acyclic graph and, hence, must have a topological order.

Compared to the reordering algorithm in {\fabricPlusplus}, ours is more fine-grained
because the unserializable transactions are aborted before ordering and the
remaining transactions are guaranteed to be serializable without being aborted.
Our reordering is no longer limited to a block's scope.
%
%
%
Another notable difference is that we determine the block snapshot at the start
of the simulation, while Fabric and {\fabricPlusplus} determine it based on the last read
operation.
We allow block commit during the contract simulation for more
parallelism, but this may introduce stale snapshots when previously read records
are updated by committed transactions during the simulation.


\revision{\subsection{Security Analysis}
\label{sec:securityanalysis}
Our reordering algorithm serves as a part of the ordering process and needs to be replicated on each honest orderer to form the ledger after the consensus service has established the transaction order. 
We assume the safety and liveness of the original consensus service under its security model, either crash-failure or byzantine-failure.
We now discuss whether both properties preserve after our reordering. 
}

\begin{algorithm}[tp]
  \caption{Contract simulation}
  \label{alg:simulation}
  \KwIn{Contract invocation context.}
  \KwOut{$readset$, $writeset$ are simulation results,\\
  \hspace*{3.4em} $b$ is the number of the block simulated on.}
  $b$ := fetch the number of the last block\;
  $readset, writeset$ := simulate the contract invocation on Block $b$
  snapshot; \hfill {// \Cref{sec:impl:snapshot_read}}
\end{algorithm}

\begin{algorithm}[tp]
  \caption{On the arrival of a transaction}
  \label{alg:txn}
  \KwData{$G$ is the transaction dependency graph with nodes $U$ and edges $V$,
    and $P$ is the pending transaction set.}
  \KwIn{$t$ is the transaction identifier, 
    $b$ is the number of the block simulated on, 
    and $readkeys$, $writekeys$ are accessed keys during simulation.}
  \KwOut{$reorderable$ property of $t$.}
  $dep$ := Compute $t$'s dependency except $c$-$ww$ among $P$ based on $G, b, readkeys,
  writekeys$; \hfill {// \Cref{sec:impl:dep}}\\
  $reorderable := true$ if no cycle is detected in $G$ with respect to $dep$,
  or $false$ otherwise;
  \hfill {// \Cref{sec:impl:graph}}\\
  \If{$reorderable$} {
    $P \unioneq \{t\}$\;
    $G.U \unioneq \{t\}$\;
    $G.V \unioneq dep$\;
  }
\end{algorithm}

\begin{algorithm}[tp]
  \caption{On the formation of a block}
  \label{alg:blk}
  \KwData{$G$ is the transaction dependency graph, and
    $P$ is the pending transaction set.}
  \KwOut{$s$ is the commit order of pending transactions.}
  $s$ := Topologically sort $P$ based on reachability in $G$\;
  $ww$ := Compute \scenario{c-ww} among $P$ with $s$\;
  $G.V \unioneq ww$;  \hfill {// \Cref{sec:impl:resintall}} \\
  $P := \emptyset$
\end{algorithm}
 
\revision{\textbf{Safety.} 
In the original Fabric design, there are four safety properties: \textit{agreement, hash chain integrity, no skipping}, and \textit{no creation}~\cite{androulaki2018hyperledger}. 
These properties require honest orderers to sequentially deliver consistent, untampered blocks in a ledger.
We claim that our approach preserves \textit{hash chain integrity} and \textit{no skipping} as we do not change the block formation procedure.
Next, \textit{no creation} holds because we do not introduce new transactions. 
Lastly, we achieve \textit{agreement} because we fully replicate the reordering on each orderer.
Moreover, we do not introduce non-determinism which may lead to execution bifurcation. 
As long as honest orderers perform the reordering individually from a consistent transaction stream, they shall produce identical ledgers. 

\textbf{Liveness.} 
Fabric defines liveness in terms of the \textit{validity} property, which mandates all broadcasted transactions to be included in the ledger. 
%
Our algorithm may compromise this liveness property as aborted transactions are excluded from the ledger. 
However, we propose the following approach to prevent abusive usage. 
To be specific, in the consensus protocol, the transaction order is tentatively proposed by a leader node. 
%
%
When this order is accepted by the other nodes, it becomes the input of our reordering approach. 
%
Hence, the order is controlled by the leader, which may hinge on the publicly available reordering algorithm to maliciously defer certain transactions.
Suppose the malicious leader detects an undesirable transaction \texttt{TxnT} which reads and writes a record against the state snapshot of block N. 
The leader, using both a proxy peer and a proxy client, can immediately prepare another transaction \texttt{TxnT'} which reads and writes the same record against block N. 
%
%
Next, the leader places \texttt{TxnT'} ahead of \texttt{TxnT} during ordering. 
The other orderers, unaware of this manipulation, may accept this ordering.
Assuming \texttt{TxnT'} passes the reorderability test in \Cref{alg:txn}, each honest orderer will abort \texttt{TxnT}.
It is because these two transactions form an unreorderable cyclic schedule, namely \texttt{TxnT'} depends on \texttt{TxnT} with \scenario{c-rw} and \texttt{TxnT} on \texttt{TxnT'} with \scenario{anti-rw}.
%
%
The crux of the mitigation is to hide the transaction's details, such as accessed records, before the transaction order is established. 
For example, we allow clients to send only the transaction hash to the orderers.
Moreover, clients have incentives to do so to avoid the above manipulation.
After the sequence of a transaction hash is decided, its details are then disclosed to orderers for reordering. 
We remark that this approach also defers malicious clients from exploiting the reordering by mutating the transaction contents.
It is because clients have already made a security commitment by publishing the transaction hash. 
}



\section{Implementation}
\label{sec:impl}

\begin{figure}[tp]
  \center
	\includegraphics[width=0.5\textwidth]{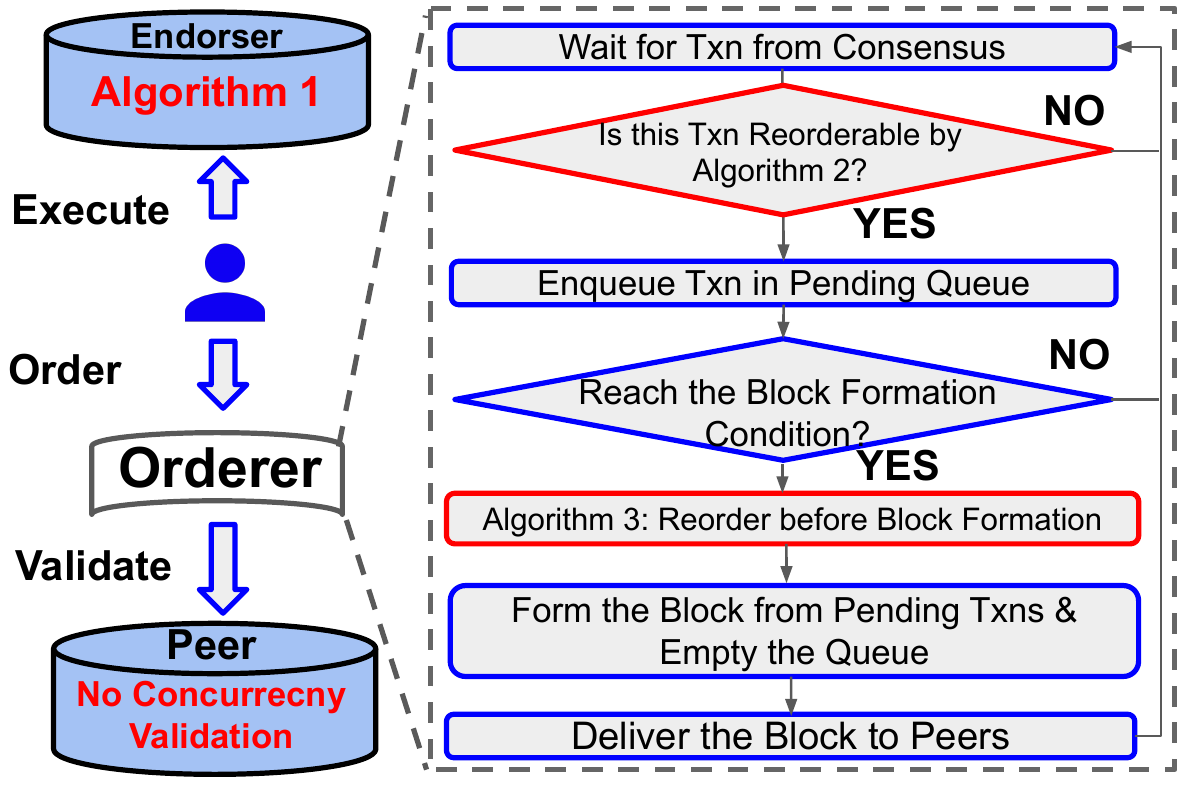}
	\caption{\revision{The integration of our approach in the ordering service}}
	\label{fig:impl:fabricx}
\end{figure}

\revision{\subsection{Overview}
We illustrate in \Cref{fig:impl:fabricx} the integration of our proposed fine-grained concurrency control in the ordering service of an EOV blockchain. 
In particular, we implemented our approach on top of both Hyperledger Fabric and FastFabric.
For simplicity, we only show a single orderer and a single peer in the EOV pipeline, but the algorithms are replicated on each node. 
While the majority of our implementation is done in orderers, \Cref{alg:simulation} is integrated in the peers for snapshot-consistent transaction execution during the endorsement phase.
In the ordering phase, we employ \Cref{alg:txn} to test the reorderability of an incoming transaction after the consensus decides its commit order.
\Cref{alg:blk} performs the abort-free reordering immediately before pending transactions are batched into a block. 
We remark that \Cref{alg:txn} and \Cref{alg:blk} are far from implementation-friendly to system developers, as both employ an abstract dependency graph.
In light of this, we present the details of designing the dependency graph and efficient operations on it.}

\subsection{Snapshot Read}
\label{sec:impl:snapshot_read}

We first describe the snapshot mechanism used by \Cref{alg:simulation}.
We rely on the storage snapshot mechanism to ensure each contract invocation is
simulated against a consistent state.
%
%
Specifically, after a block is committed, we create a storage snapshot and associate
it with the block number.
Each transaction, before its simulation, must acquire the number of the latest
block, as shown in \Cref{alg:simulation}.
%
%
Staled snapshots without any simulation are periodically pruned.
%
%
This design allows more parallelism across contract simulation in the Execution
phase and block commit in the Validation phase.
In contrast, vanilla Fabric uses a read-write lock to coordinate these two
phases.

\subsection{Dependency Resolution}
\label{sec:impl:dep}

To compute the dependency graph in \Cref{alg:txn}, we introduce two
multi-versioned storages in the orderers to identify committed transactions.
These storages are implemented in LevelDB and represent
\textit{CommittedWriteTxns} (\textit{CW}) and \textit{CommittedReadTxns}
(\textit{CR}), respectively.
Each key of \textit{CW} consists of the concatenation of the record key and the
commit sequence of the transaction updating the value.
For example, if \texttt{Txn1} with commit sequence $(3,2)$ writes to key
\texttt{A}, \textit{CW} has an entry \dbentry{A}{3}{2}{Txn1}.
Similarly, each key of \textit{CR} consists of the concatenation of the record key and
the commit sequence of the transaction reading that key's latest value.
For instance, the entry \dbentry{A}{4}{1}{Txn7} indicates that \texttt{Txn7} is
the first transaction in block 4 which reads the latest value of key \texttt{A}.
In both \textit{CW} and \textit{CR}, we place the record key prior to the commit
sequence to efficiently support point query and range query.
For example, the query $CW$.$Before(key, seq)$ returns the last committed
transaction updating $key$ with the commit sequence earlier than $seq$.
Similarly, $CW$.$Last(key)$ returns the last committed transaction updating
$key$.
For the range query, $CW[key][seq:]$ returns all committed transactions from
$seq$ onward that update $key$.

We maintain two in-memory indices, $PendingWriteTxns$ ($PW$) and
$PendingReadTxns$ ($PR$), to respectively store the keys for the write and read
sets of pending transactions.
Consider a new transaction $txn$ that starts at $startTS$ with read keys $R$ and
write keys $W$. All the dependencies of transaction $txn$ are computed as
follows.

\begin{center}
  \def\arraystretch{1.3}
  \begin{tabular}{lcl}
    $\scenario{anti-rw}(txn)$
    & $=$
    & $\bigcup_{r \in R}^{} CW[r][startTS:] \cup PW[r]$
    \\

    $\scenario{rw}(txn)$
    & $=$
    & $\bigcup_{w \in W}^{} CR[w]  \cup PR[w]$
    \\

    $\scenario{n-wr}(txn)$
    & $=$
    & $\bigcup_{r \in R}^{} CW.Before(r, startTS)$
    \\

    $\scenario{ww}(txn)$
    & $=$
    & $\bigcup_{w \in W}^{} CW.Last(w)$
    \\
  \end{tabular}
\end{center}

Note that we ignore \scenario{ww} dependencies between pending transactions and do not
differentiate whether \scenario{ww} and \scenario{rw} are concurrent or not.
This is because non-concurrent transaction may be part of a cycle.
We then compute the predecessor transactions of $txn$ as $\scenario{ww}(txn)
\cup \scenario{n-wr}(txn) \cup \scenario{rw}(txn)$, and successor transactions
as $\scenario{anti-rw}(txn)$.

\subsection{Cycle Detection}
\label{sec:impl:graph}

We now discuss how we represent the dependency graph $G$ to detect cycles and
achieve serializability.
%
%
We face two design choices.
On the one hand, we could maintain only the immediate linkage information for
each transaction and then perform graph traversal for cycle detection.
On the other hand, we could maintain the entire reachability information among
each pair of transactions.
But the latter approach shifts the overhead from computation to space
consumption.
%
%
We achieve a sweet spot by maintaining the immediate successors of
a transaction ($txn.succ$) and represent all transactions that can reach $txn$
with a bloom filter, referred to as $txn.anti\_reachable$.
Cycle detection becomes straightforward by testing
$p.anti\_reachable(s)$ for each pair $(p,s)$ consisting of a
predecessor and a successor of $txn$.

We use bloom filters because they are memory efficient and can perform
fast union.
Union is extensively used to update the reachability information for each
transaction, as shown in \Cref{alg:reachability}.
Since a bloom filter internally relies on a bit vector, the set union can be
fast computed via the bitwise OR operation.
However, bloom filters are known to report false positives~\cite{bloom1970space}.
%
%
If the filters report such false positives for a pair of adjacent transactions
to $txn$, we preventively abort $txn$.
If they report negative for all pairs, then $txn$ does not belong to any
cycle in $G$.

\Cref{alg:reachability} entails the relatively expensive traversal of all
reachable transactions from $txn$.
However, this cost is bearable, since the traversal is unnecessary when
$\scenario{anti-rw}(txn)$ is empty.
This is often the case under non-skewed workloads.
%
Moreover, we reduce the cost of traversal by pruning the dependency graph, as described in \Cref{sec:impl:optimization}.




\begin{algorithm}[tp]
    \caption{Reachability update for transaction $txn$}
    \label{alg:reachability}
    \KwData{$G$ is the transaction dependency graph}
    \KwIn{$M$ is the number of next block to be committed, 
      $pred$ is $txn$'s immediate predecessor transactions, and
      $succ$ is $txn$'s immediate successor transactions.}

    $txn.anti\_reachable$ := $\emptyset$\;
    \For{$p$ \textup{in} $pred$} {
        $p.succ \ \unioneq \{txn\}$\;
        $txn.anti\_reachable \ \unioneq
        p.anti\_reachable$\; }
    \For{$s$ \textup{reachabale from} $succ$ \textup{in} $G$} {
        $s.anti\_reachable \ \unioneq
        txn.anti\_reachable$;\\
        $s.age := M$;\label{algo:age}}
\end{algorithm}

Another issue in \Cref{alg:reachability} is the constant growth of the
$anti\_reachable$ filter.
In practice, we observe that the false positive rate of a single bloom filter
grows to an intolerable ratio.
To address this issue, we use two bloom filters with relay.
Each transaction is associated with one bloom filter capturing transactions
committed after block $M$ and another bloom filter capturing transactions after
block $N$.
Suppose block $C$ is the earliest block which contains a committed transaction
in $G$.
We maintain $M < C < N$ and use the first bloom filter for testing reachability.
Whenever $C$ grows to $M < N < C$, the first bloom filter is emptied and it
starts to collect transactions from the current block.
We then use the second filter for testing reachability.
In this manner, we restrict the number of transactions represented by a bloom
filter within a certain block range so that the false positive rate remains
acceptable. 
\revision{For safety, honest orderers must use the same $M$ and $N$ for exact replication.}

\subsection{Dependency Restoration}
\label{sec:impl:resintall}


Next, we present our method to install \scenario{ww} dependencies into the dependency
graph $G$ based on the derived commit sequence, which is a topological order of
the pending transactions $P$ according to the reachability in $G$.
One prominent issue is that the reachability of a transaction may be affected by
multiple \scenario{ww} dependencies from various updated keys.
But we want the reachability modification to take place within a single
iteration for efficiency.
\Cref{alg:restoration} outlines the major steps of the restoration of
\scenario{ww} dependencies.
We further explain this algorithm using the example in \Cref{fig:impl:restore}.

\begin{algorithm}[tp]
  \caption{Restoration of \scenario{ww} within pending transactions based on
    the computed commit sequence}
    \label{alg:restoration}
    \KwData{$G$ is transaction dependency graph.}
    \KwIn{$seq$ is committed sequence of pending transactions, 
      $PW$ is the index that associates updated keys with pending
      transactions.}
    $head\_txns := \emptyset$\;
    \For{($key, txns$) \textup{in} $PW$\label{algo:iter_key}}
      {
        Sort $txns$ based on the relative order in $seq$;\\
        $(txn1, txn2)$ := the first pair in $txns$ such that $txn1 \notin
        txn2.anti\_reachable$;\label{algo:pair}\\
        $txn2.anti\_reachable \ \unioneq
        txn1.anti\_reachable$;\\
        $head\_txns \ \unioneq \{txn2\}$;
      }
    \For{$txn$ \textup{in the topologically-ordered
         iteration of all txns reachable from} $head\_txns$
         \label{algo:iteration}}
      {
        \For{$t$ \textup{in} $txn.succ$}
          {
            $t.anti\_reachable \ \unioneq txn.anti\_reachable$;
          }
      }
\end{algorithm}

For each $key$ to be updated by pending transactions ($PW$), we topologically
sort its associated transactions and select the first pair that is not yet
connected in the reachability filter.
In such a pair, the second transaction can be reached from all the predecessors
of the first transaction.
There can be a scenario where transactions in a pair are already connected in
the reachability filter, which makes the restoration redundant.
For example, this happens with \texttt{Txn0} and \texttt{Txn3} in
\Cref{fig:impl:restore}.
For transactions that are not yet connected, we need to update their
successors.
To do this efficiently, we keep the transactions in a set ($head\_txns$) and
update their successors based on the topological order.
Thereby, we avoid updating the information multiple times during the iteration
in line~\ref{algo:iter_key}.
For example, \texttt{Txn8} in \Cref{fig:impl:restore} is reachable through
the update of both key \texttt{A} and \texttt{B}.
Using our algorithm, the reachability information is updated once.


\begin{figure}[tp]
  \center
	\includegraphics[width=0.5\textwidth]{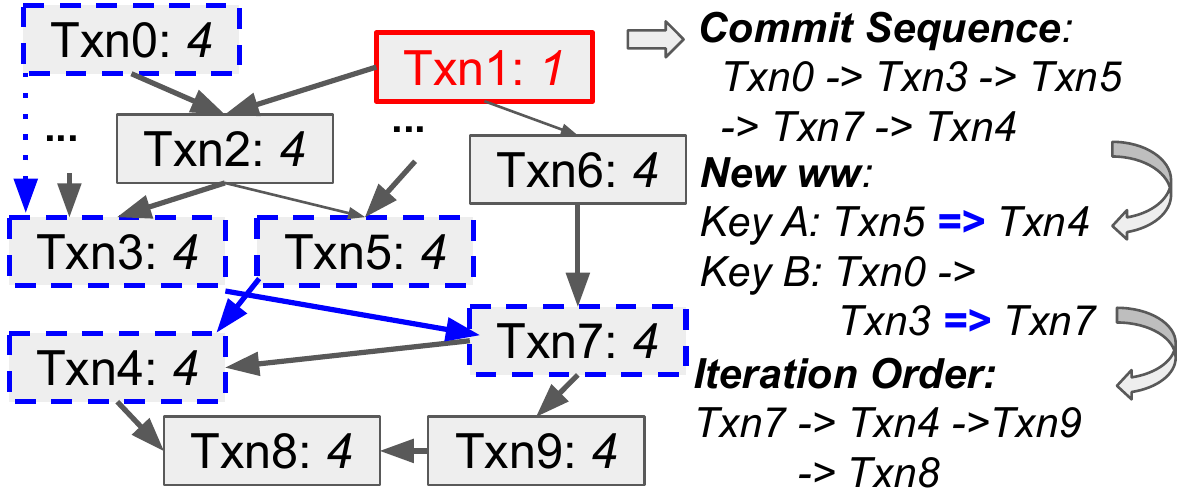}
	\caption{Example of a dependency graph with \textcolor{blue}{pending
      transactions} (marked with blue dashed border), the commit sequence,
    \textcolor{blue}{new \scenario{ww} dependencies}
    (marked with blue solid line) and the topologically-sorted iteration order.
    We do not consider the \scenario{ww} dependency between Txn0 and Txn3
    (marked with blue dotted line), as it is implicit.
    \textcolor{red}{Txn1} in red is subject to
    pruning due to staleness. The transaction age is in italic.}
	\label{fig:impl:restore}
\end{figure}

\subsection{Dependency Graph Pruning}
\label{sec:impl:optimization}


Since graph $G$ can grow quickly, we prune transactions that either (i) are simulated against very old snapshots or (ii) cannot affect pending transactions.
For the first case, we introduce a parameter called $max\_span$ to limit the block span\footnote{If a transaction is simulated against block $M$ and committed in block $M+1$, its block span is 1.} for
a transaction.
If the number of the next block is $M$, we compute the \textit{snapshot
  threshold} as $H = M - max\_span$.
Any transaction simulated against block $H$ or earlier is aborted.
For the second case, we define the \textit{age} of a transaction $txn$ to be the sequence
number of the last committed block containing at least one transaction reachable
from $txn$ in $G$.
When the snapshot threshold is greater than $txn$'s age, future transactions
cannot be concurrent with any transaction that can reach $txn$.
In this case, the \scenario{anti-rw} dependency will not happen, and this rules
out any unserializable schedule containing $txn$.
Therefore, $txn$ can be safely pruned from $G$.
We facilitate the pruning by arranging all transactions in $G$ into a priority
queue weighted by age.
For new transaction to be committed in block $M$, we increase the age of the
transactions reachable from it to $M$ during the traversal in
\Cref{alg:reachability} (line~\ref{algo:age}).
\revision{For security, all orderers must use the same value for $max\_span$.}





\newcommand{\foccs}{Focc-s}
\newcommand{\foccl}{Focc-l}

\section{Experiments}
\label{sec:experiment}

%
%
\subsection{Systems and Setup}

\revision{
First, we implement our approach on top of Hyperledger Fabric 1.3 and name the resulting system \textit{{\fabricSharp}}.
}
We compare {\fabricSharp} with the vanilla Fabric \cite{androulaki2018hyperledger}, {\fabricPlusplus} \cite{sharma2019blurring},
and two new systems which we developed by directly adopting OCC techniques from databases to Fabric. 
The first system, which we call \textit{\foccs}, follows the standard serializable OCC approach in ~\cite{CahillRF08}. 
This approach considers a dangerous pattern formed by two consecutive concurrent read-write conflicts with at least one \textit{anti-rw}. 
We modify our \Cref{alg:txn} such that incoming transactions with a \textit{c-ww} conflict or a dangerous pattern are immediately aborted. 
{\foccs} does nothing on block formation.
The second system is \textit{\foccl}, which uses a recent OCC technique~\cite{ding2018improving}, based upon which we construct the read-write
dependency graph and apply its \textit{Sort-Based Greedy Algorithm} in
\Cref{alg:blk} for reordering. 
{\foccl} does not filter any transactions in \Cref{alg:txn}.
\revision{
Second, we implement our fine-grained concurrency control on top of FastFabric [20], an orthogonal Fabric improvement, and name the resulting system \textit{{\ffabricSharp}}.
We aim to investigate the performance of {\ffabricSharp} on a real production workload with low contention.
In both sets of experiments, we deploy two \textit{orderers}, three Kafka nodes and four \textit{peers}, each on a physical machine with Intel Xeon E5-1650 3.5GHz CPU and 32GB RAM. 
The machines are connected via 1Gb Ethernet.
We configure the smart contract to be endorsed (executed) by a single peer.
Any of the four peers can serve as the endorser to spread the workload. 
}
The experiments are run at least three times and the average values are reported.

\subsection{Workloads and Benchmark Driver}
We use the same workloads that evaluate {\fabricPlusplus} \cite{sharma2019blurring},
which are based on the Smallbank benchmark.
\revision{
A transaction reads and writes 4 bank accounts, respectively, out of 10k accounts.
We set 1\% of them as hot accounts.}
Each read has a certain probability to access the
hot accounts, controlled by the \textit{Read hot ratio} parameter. 
Similarly, writing to hot accounts is controlled by the \textit{Write hot ratio}. 
We introduce two more workload parameters, namely \textit{Client
  Delay} and \textit{Read Interval}.
The former controls the delay of a client's broadcast to \textit{orderer} after
it receives the execution results from a \textit{peer}.
This parameter simulates the network transmission delay at the client side. 
The latter simulates computation-heavy transactions by controlling the interval
between consecutive reads. 
\Cref{tab:parameter} tabulates all the parameters, with the default value underlined. 
\revision{
We fix $max\_span$ to 10 and the request rate to 700 tps.
This is because Fabric can sustain a maximum raw throughput of around 700 tps on our setup, as shown in \Cref{fig:intro}.
Unless otherwise specified, all reported throughputs denote the \textit{effective throughput}, which represents the transactions that pass the serializability check and persist their states.}


\begin{table}
	\centering
	\caption{Experiment parameters}

	\label{tab:parameter}
	\begin{tabular}{@{}ll@{}}
	\toprule
	\textbf{Parameter}
  & \textbf{Value} \\

	\midrule

	\# of transactions per block
  & 50, 100, \underline{200}, 300, 400, 500 \\

  Write hot ratio (\%)
  & 0, \underline{10}, 20, 30, 40, 50 \\

	Read hot ratio (\%)
  & 0, \underline{10}, 20, 30, 40, 50 \\

  Client delay (x100 ms)
  & \underline{0}, 1 ,2 ,3, 4, 5 \\

	Read interval (x10 ms)
  & \underline{0}, 4, 8, 12, 16, 20 \\

	\bottomrule
	\end{tabular}
\end{table}

\subsection{The Performance of \fabricSharp}


\begin{figure}[tp]
	\centering
    \begin{subfigure}{0.40\textwidth}
      \includegraphics[width=0.8\textwidth]{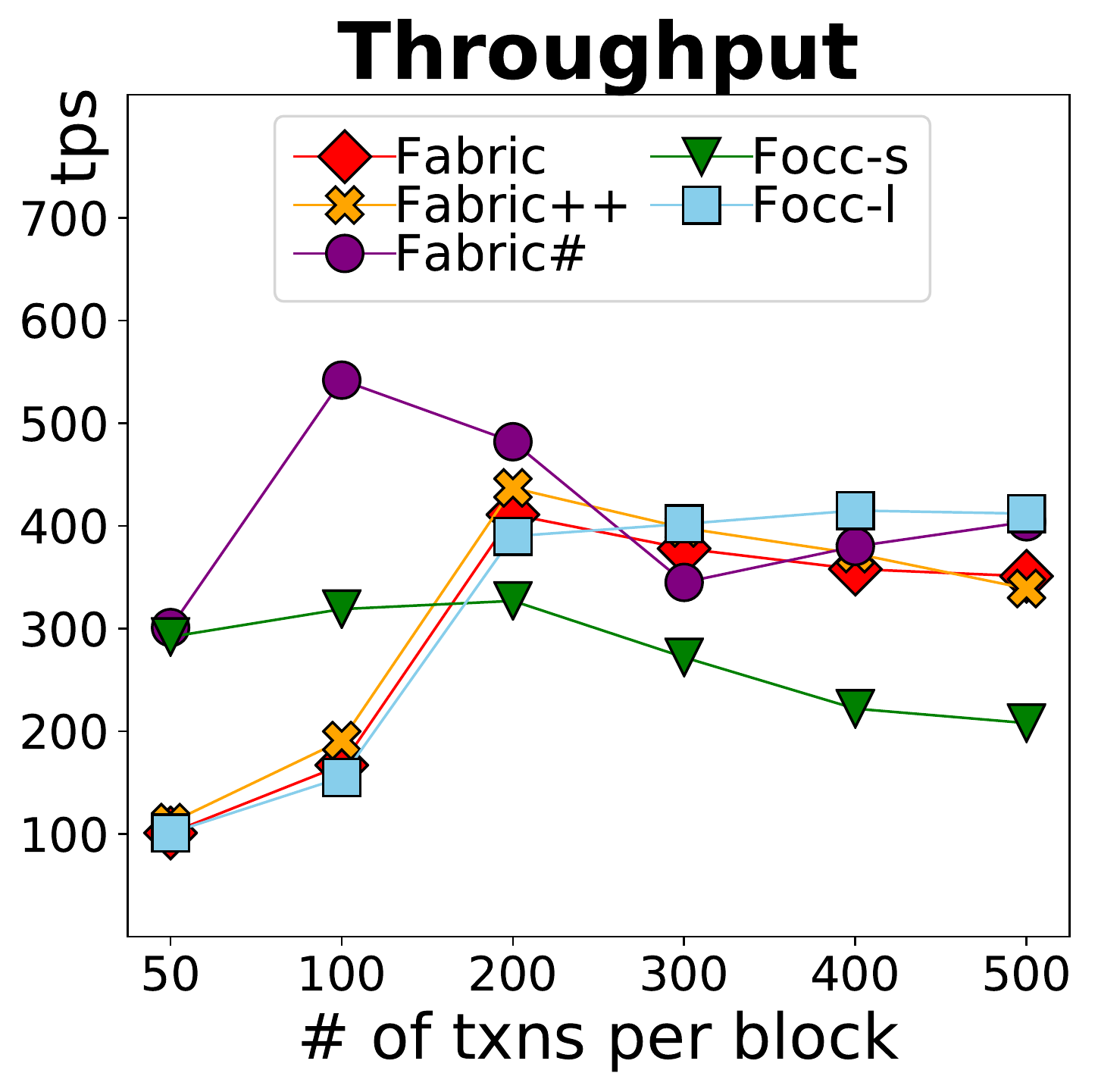}
    \end{subfigure}
    \begin{subfigure}{0.40\textwidth}
      \includegraphics[width=0.8\textwidth]{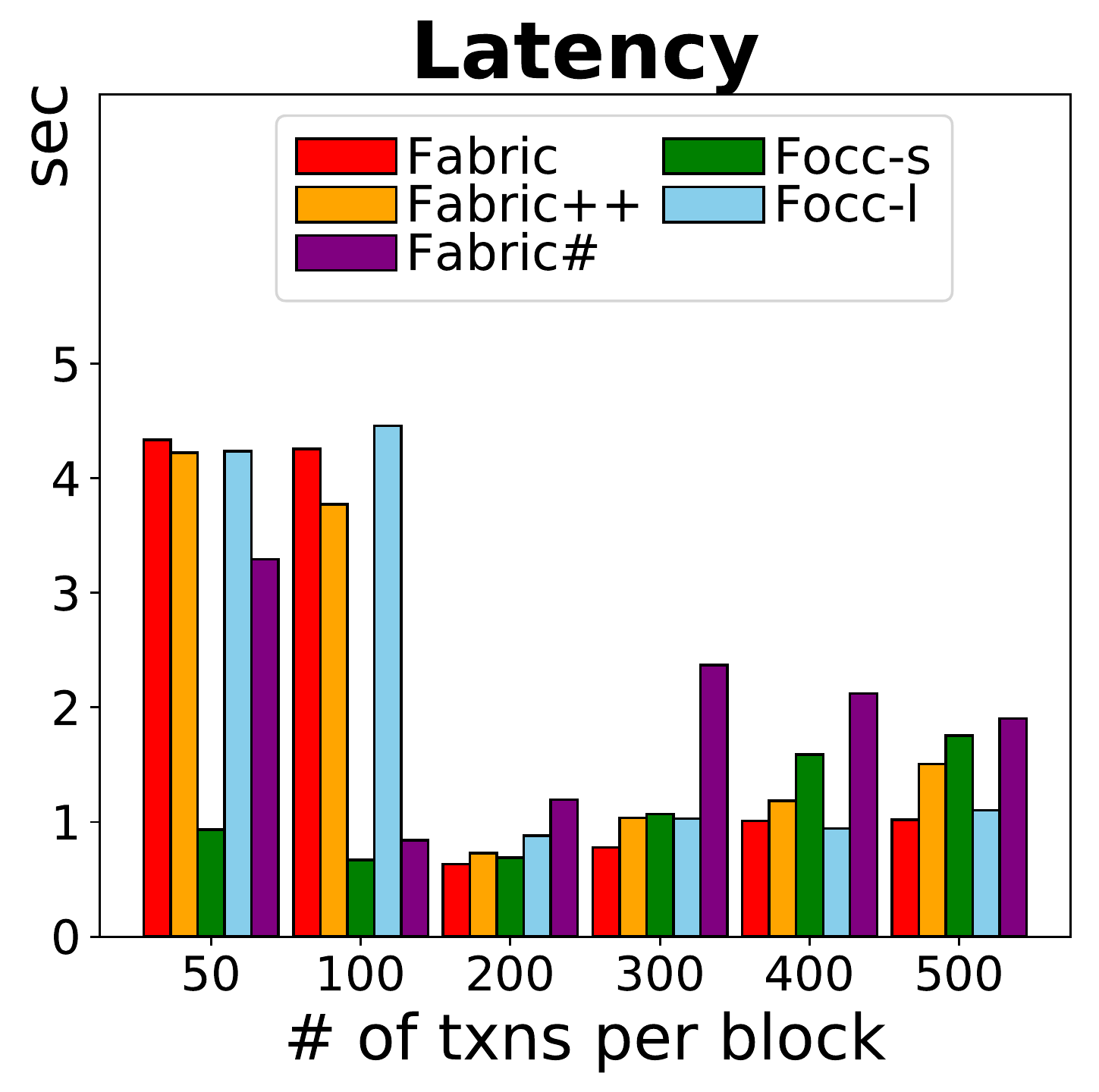}
    \end{subfigure}

    \caption{Performance under varying block size}
    \label{chart:blk_size}

\end{figure}

\textbf{Block Size.}
\revision{We first determine the block size that leads to the highest throughput for each system, and use these sizes for the remainder of the experiments.
\Cref{chart:blk_size} shows that the highest throughput (542 tps) is achieved by {\fabricSharp} when the block size is set to 100 transactions.
In contrast, Fabric, {\fabricPlusplus}, {\foccs} and {\foccl} reach their peak performance, 411, 437, 327, and 415 tps, respectively, when a block is limited to 200, 200, 200, and 400 transactions, respectively. 
Hence, our {\fabricSharp} achieves 25\% improvement in throughput compared to the state-of-the-art {\fabricPlusplus}.
%
%

Contrary to our expectation, {\fabricPlusplus} does not achieve higher throughput with larger blocks, even though there are more transactions available for reordering. 
We attribute it to the longer latency that intensifies the contention, leading to more unserializable transactions. 
On the one hand, it takes longer to form a block.
On the other hand, the reordering before block formation requires more time because there are more transactions.
For example, we observe that reordering in {\fabricPlusplus} takes 4.3ms with 50 transactions per block and 401ms with 500 transactions per block. 
In contrast, {\foccl} takes 0.12ms and 5.19ms, respectively, due to its light-weight approach, even though it similarly constructs a dependency graph.
This explains why {\foccl} has shorter delay and performs better on larger blocks. 
But when the block size is smaller than 200, the throughput of both {\fabricSharp} and {\foccs} is significantly higher compared to Fabric, {\fabricPlusplus}, and {\foccl}.
It is because of the preventive abort applied by both {\fabricSharp} and {\foccs} during the Ordering phase before reordering.
This alleviates the congestion in the Validation phase, which is the bottleneck.
In contrast, Fabric, {\fabricPlusplus}, and {\foccl} exhibit high latency with smaller blocks because many unserializable transactions are included in the ledger and they overload the Validation phase.
}

\begin{figure}[tp]
	\centering
    \begin{subfigure}{0.40\textwidth}
      \includegraphics[width=0.8\textwidth]{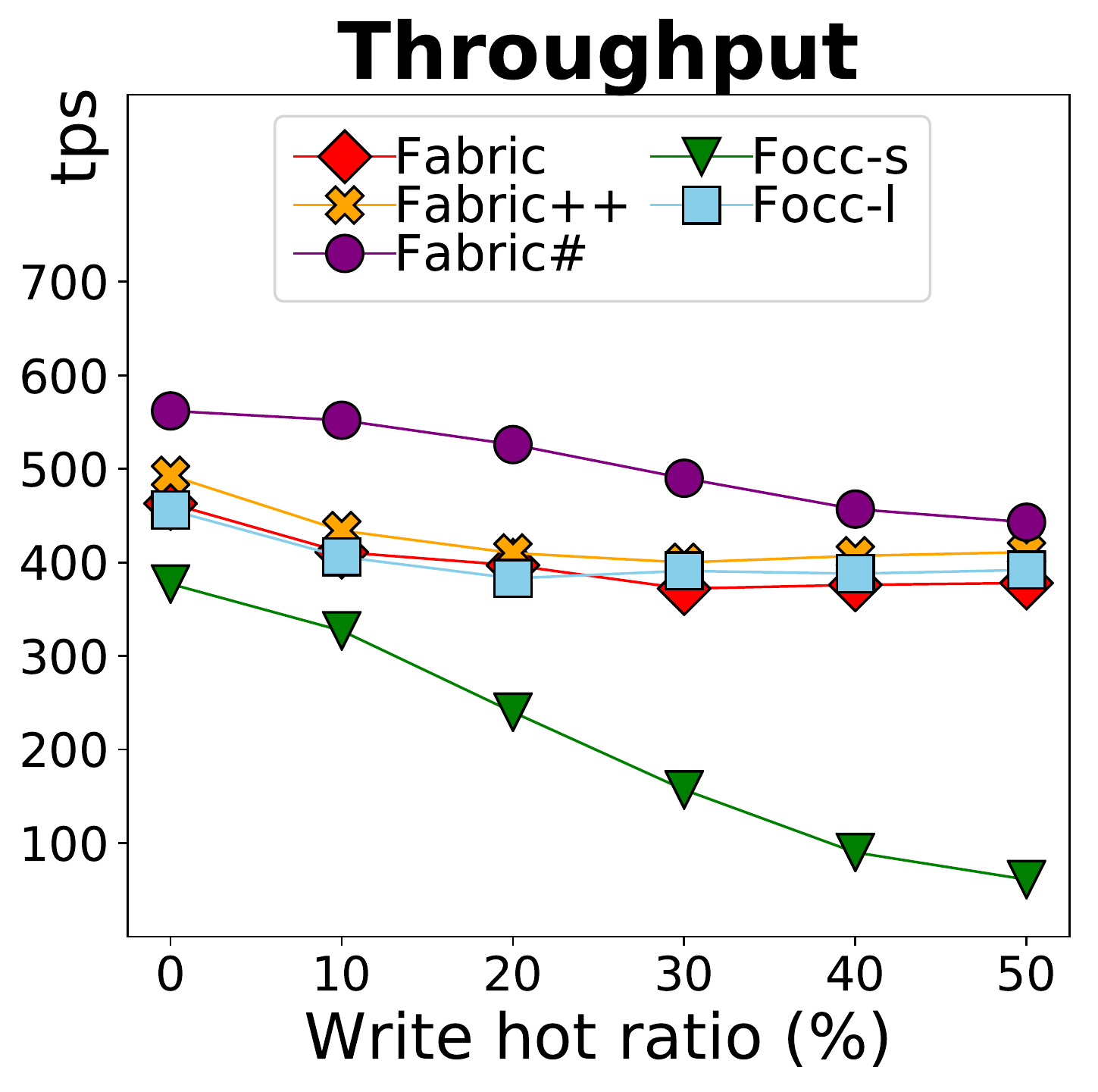}
      \label{chart:write_hot_thruput}
    \end{subfigure}
    \begin{subfigure}{0.40\textwidth}
      \includegraphics[width=0.8\textwidth]{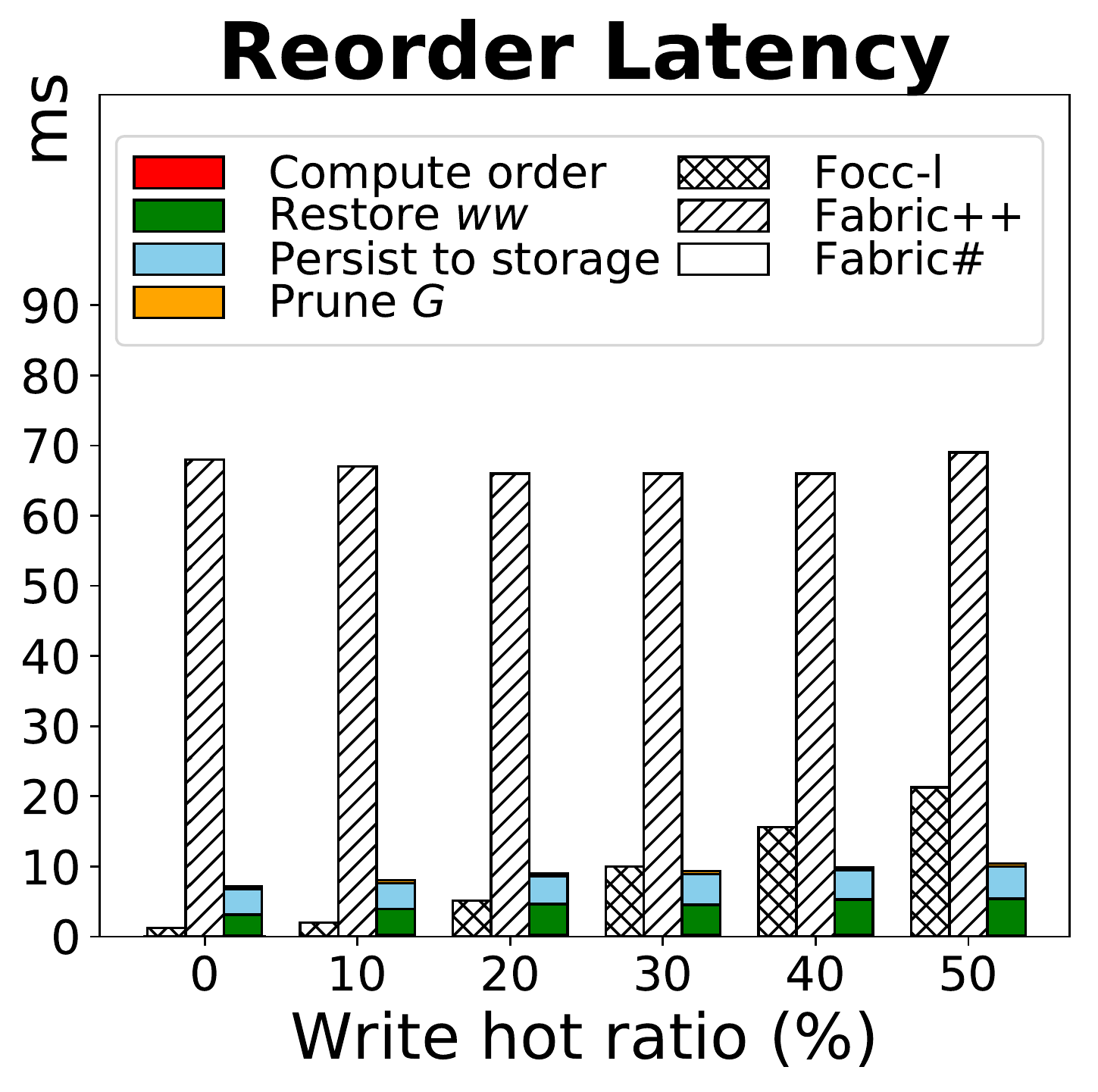}
      \label{chart:write_hot_blk_delay}
    \end{subfigure}

    \caption{Throughput and reordering latency under varying write hot ratio}
    \label{chart:write_hot}

\end{figure}

\textbf{Write Hot Ratio.}
To evaluate the effect of write-write conflicts, we concentrate more write operations into a fixed number of hot accounts.
The throughput of {\fabricSharp} remains the highest among all the systems, as shown in \Cref{chart:write_hot} (left). 
As expected, the throughput of {\foccs} drops significantly due to its prevention on \textit{c-ww}. 
\revision{
We also observe that the reordering latency of {\fabricPlusplus} is constantly large, 
while this delay in {\foccl} is smaller and proportional to the increasing skewness.
This is because {\foccl} iterates through the dependency graph of pending transactions in rounds. 
In each round, its \textit{Sort-Based Greedy Algorithm} keeps pruning transactions until there are only transactions without dependencies.
In contrast, {\fabricPlusplus} computes all the cycles and determines the transactions to be aborted in batch mode. 
Hence, its reordering procedure is less sensitive to workload skewness compared to {\foccl}. 
\Cref{chart:write_hot} shows that the reordering latency in {\fabricSharp} (\Cref{alg:blk}) is low.
This is because {\fabricSharp} shifts most of the work (e.g., the dependency graph maintenance) to \Cref{alg:txn} on the transaction arrival. 
We notice that a large ratio (\~50\%) of the reordering delay in {\fabricSharp} is for the restoration of \textit{ww} conflicts, and this ratio increases with higher write hot ratio. 
%
}
\begin{figure}[tp]
	\centering
    \begin{subfigure}{0.40\textwidth}
      \includegraphics[width=0.8\textwidth]{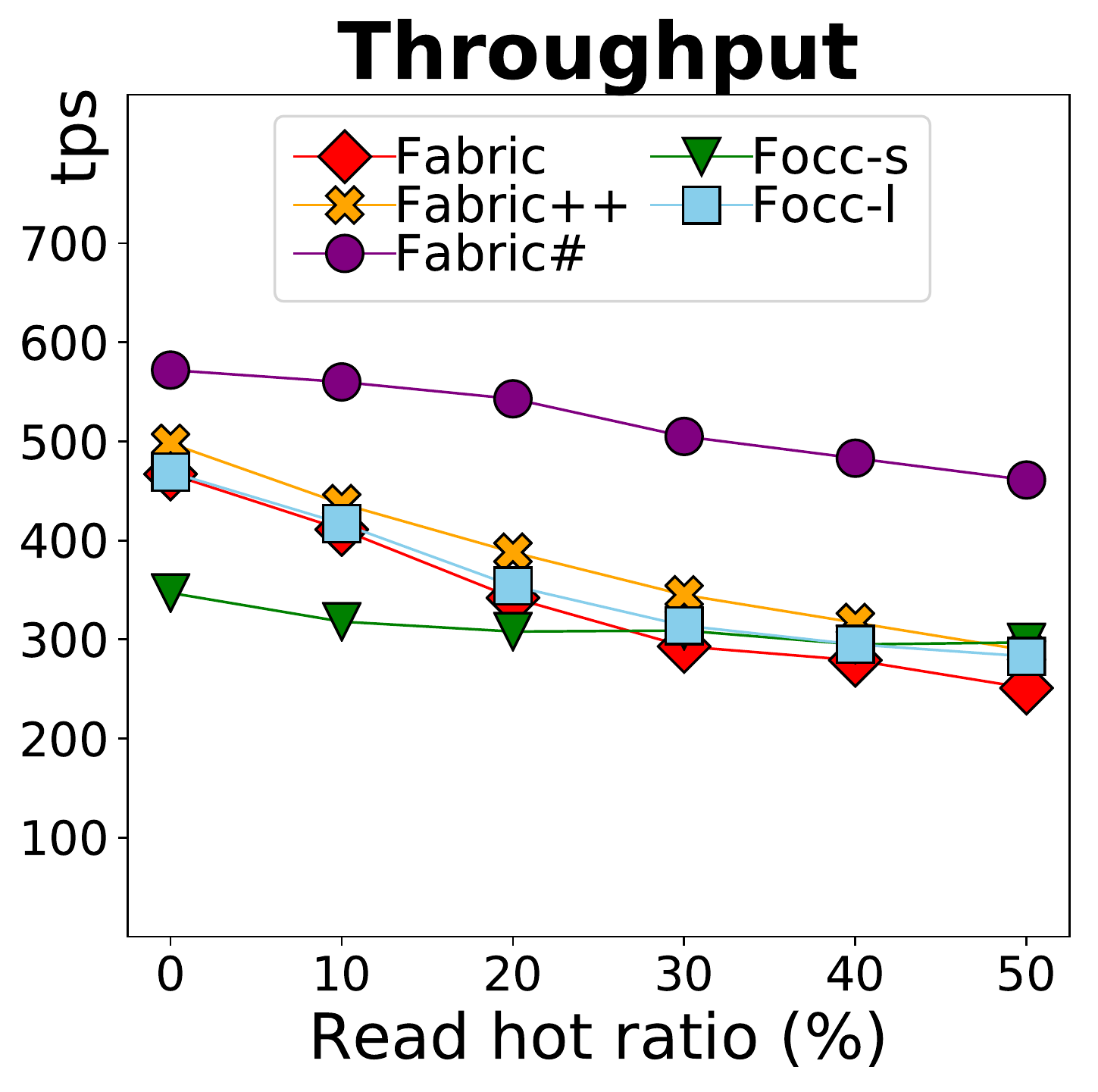}
    \end{subfigure}
    \begin{subfigure}{0.40\textwidth}
      \includegraphics[width=0.8\textwidth]{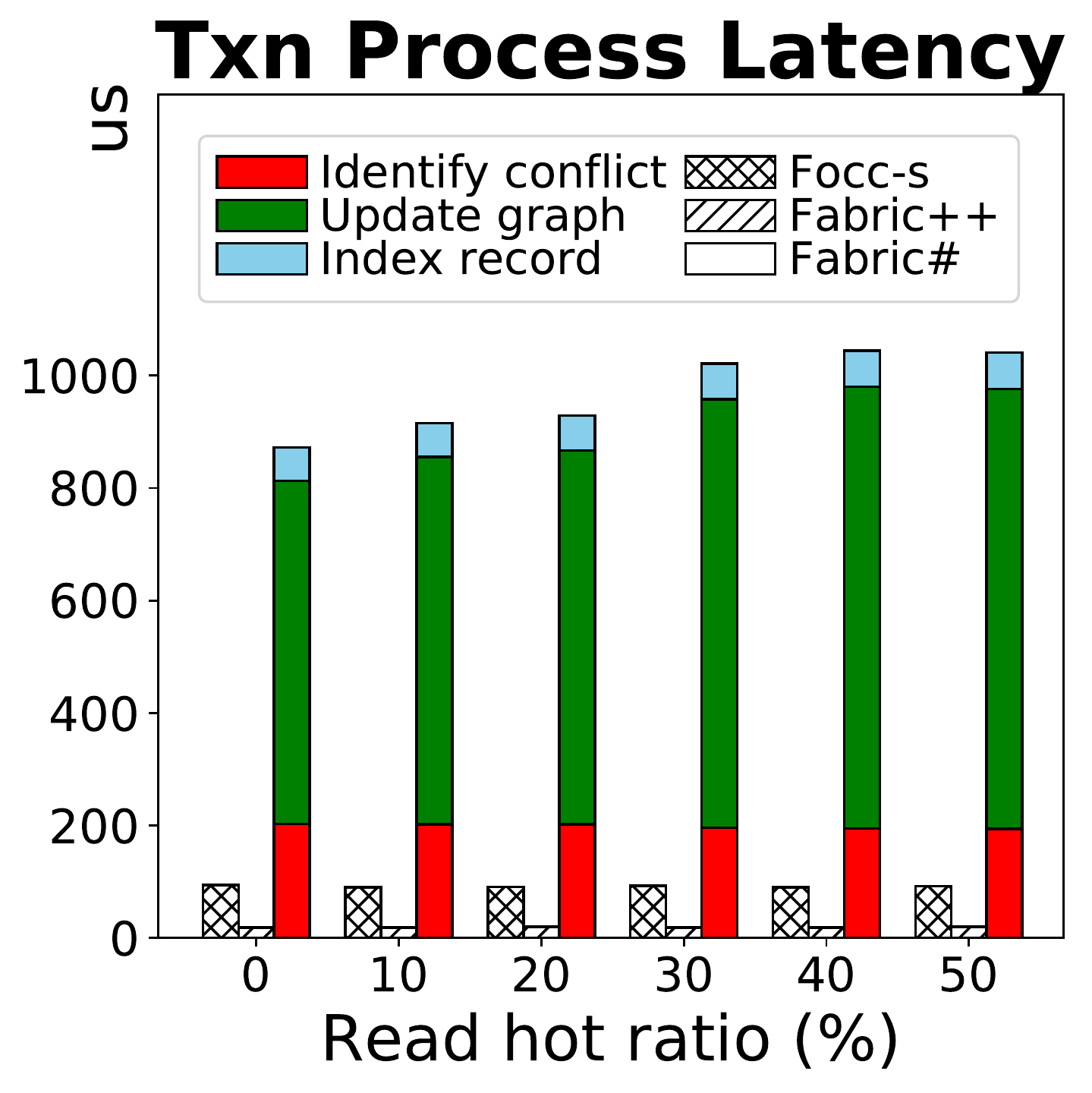}
    \end{subfigure}

    \caption{Throughput and transaction processing latency under varying read
      hot ratio}

    \label{chart:read_hot}
    
    \vspace{-10pt}
\end{figure}

\textbf{Read Hot Ratio.}
We increase the read hot ratio to generate more read-write conflicts in the workload. 
As explained in \Cref{theory:unreorderable}, dependency cycles with 
these conflicts can never be reordered to become serializable. 
Consistent to our explanation, we show in \Cref{chart:read_hot} that the throughput of all the systems, except {\foccs}, decreases at a similar rate. 
The throughput of {\foccs} is greater compared to Fabric and {\foccl} when 50\% of the read requests are on the hot accounts. 
This is because {\foccs} imposes a more stringent condition for serializability compared to Fabric and {\foccl}.
{\foccs} aborts transactions if they are forming two consecutive read-write conflicts with at least one \textit{anti-rw}, while the other systems, except {\fabricSharp}, abort immediately when there is a single \textit{anti-rw}. 
Hence, {\foccs} can recover more serializable transactions especially under heavy read-write contention. 
\Cref{chart:read_hot} (right) shows the processing latency breakdown for an incoming transaction.
As expected, the reachability update on the dependency graph takes the largest proportion of the delay in {\fabricSharp}, as all the reachable transactions from the incoming one must be traversed.
This overhead increases with more dependencies in the workload. 
\revision{
Compared to {\fabricSharp}, the transaction processing delay in {\foccs} and {\fabricPlusplus} is almost negligible.
However, {\foccs} takes a bit longer than {\fabricPlusplus}, as it needs to additionally identify conflicted transactions, instead of only indexing the transactions based on the accessed records as in {\fabricPlusplus}. 
}

\begin{figure}[tp]
	\centering
    \begin{subfigure}{0.40\textwidth}
      \includegraphics[width=0.8\textwidth]{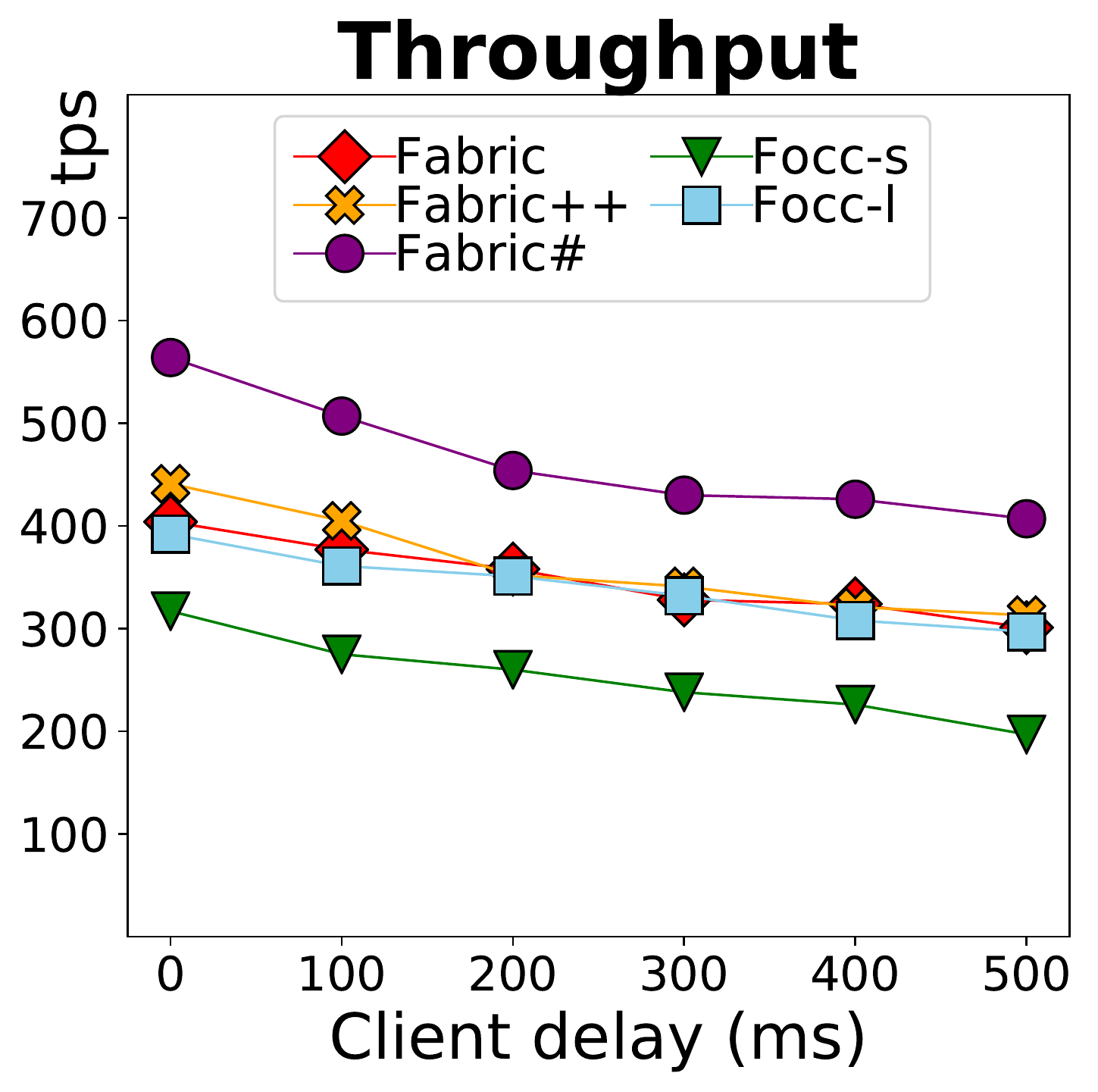}
    \end{subfigure}
    \begin{subfigure}{0.40\textwidth}
      \includegraphics[width=0.8\textwidth]{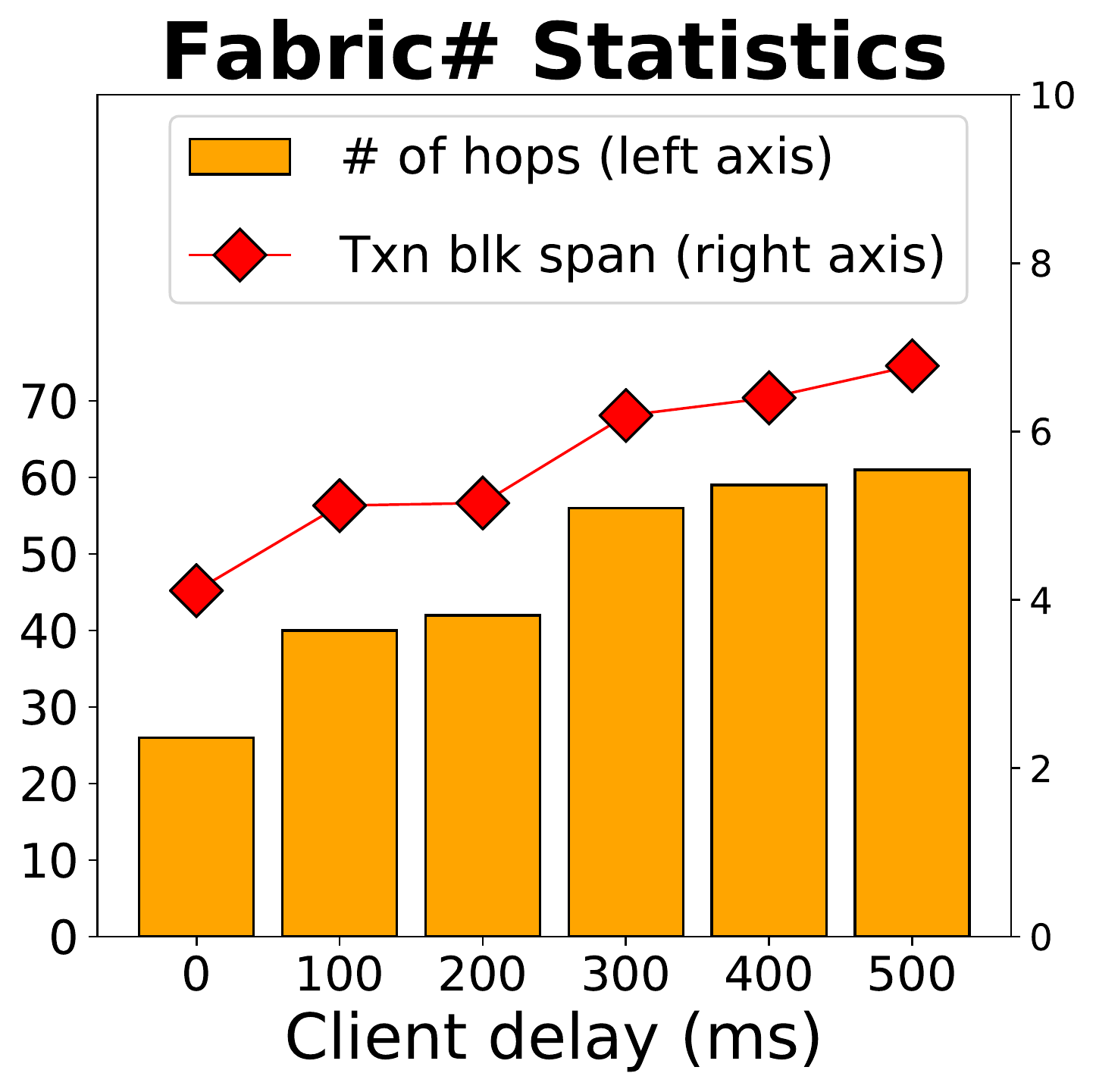}
    \end{subfigure}

    \caption{Throughput of all systems (left) and statistics of {\fabricSharp} (right) under varying client delay}
    \label{chart:client_delay}

\end{figure}
\textbf{Client Delay.}
Next, we simulate the network transmission latency at the client side, in order to study its impact on a transaction's end-to-end processing.
Using the \textit{client delay} parameter, we introduce a delay between the Execution and Ordering phases.
As expected, a longer client delay increases the end-to-end latency and the block span of a transaction.
In turn, this leads to lower throughput, as show in \Cref{chart:client_delay} (left).
Moreover, a larger block span leads to more concurrent transactions and more dependencies.
As shown in \Cref{chart:client_delay} (right), {\fabricSharp} traverses more transactions in the dependency graph to update their reachability (\Cref{alg:txn}) when the client delay is higher.
Despite this, {\fabricSharp} performs better than all the other systems.


\begin{figure}[tp]
	\centering
    \begin{subfigure}{0.40\textwidth}
      \includegraphics[width=0.8\textwidth]{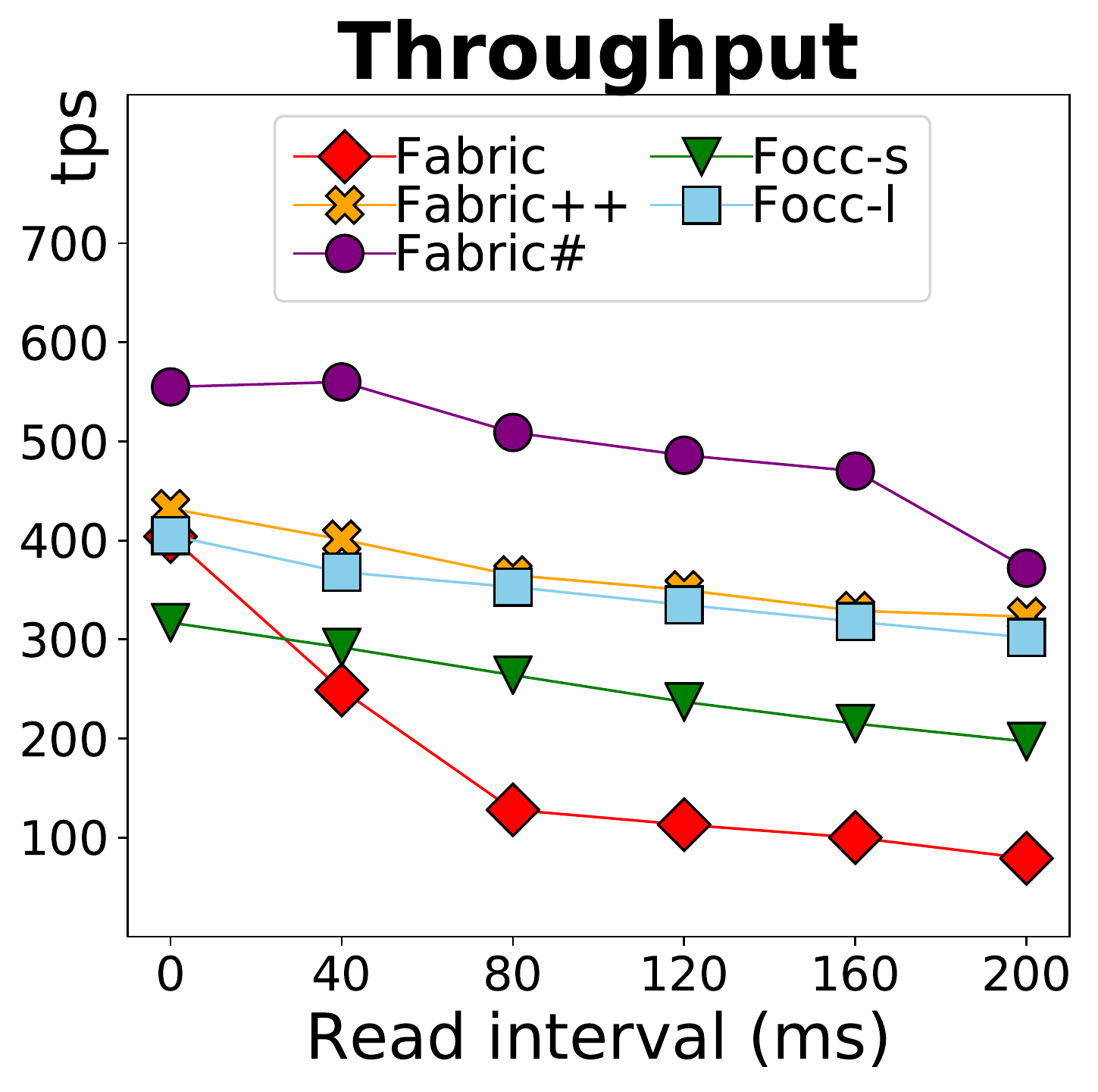}
    \end{subfigure}
    \begin{subfigure}{0.40\textwidth}
      \includegraphics[width=0.8\textwidth]{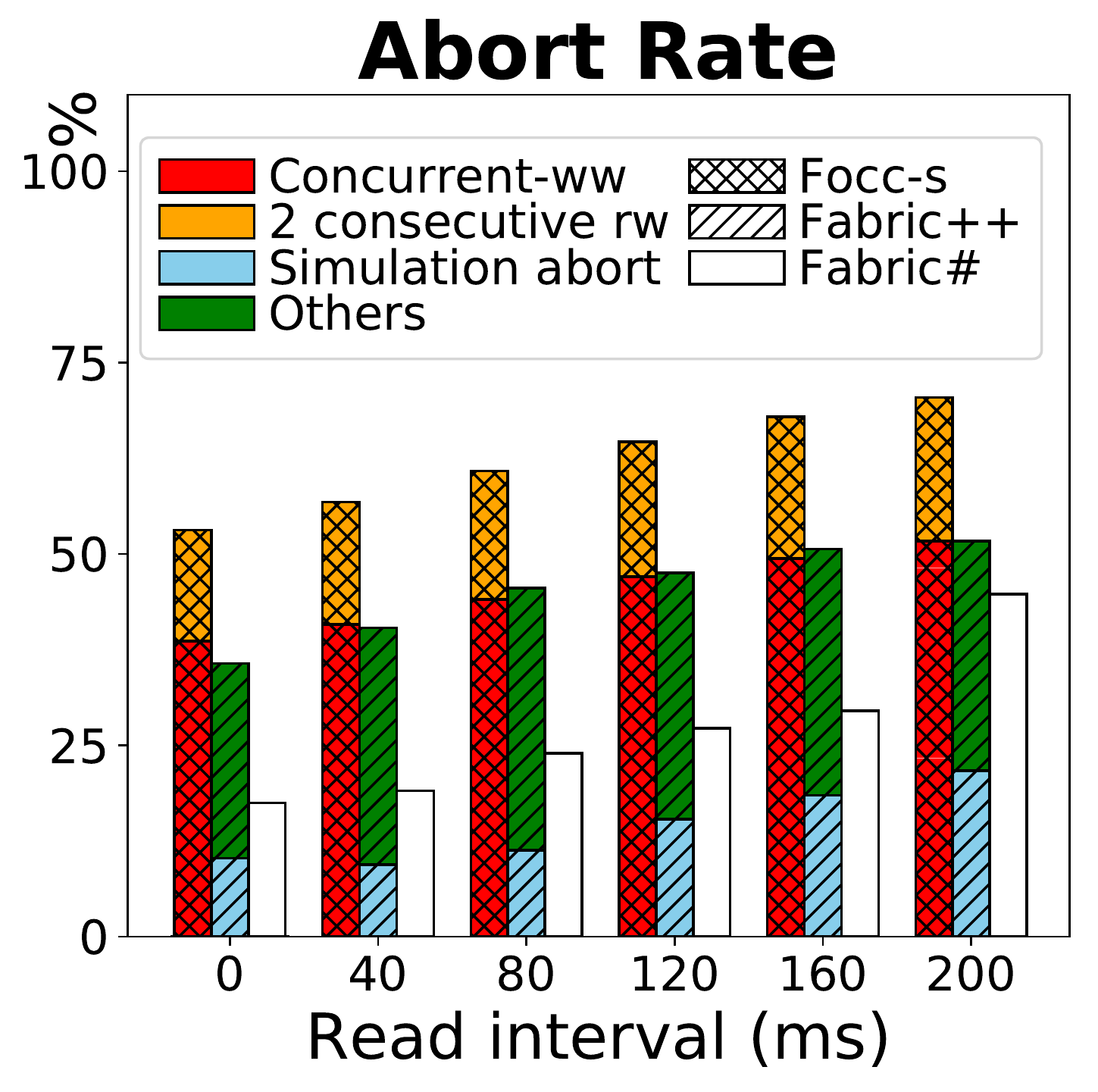}
    \end{subfigure}

    \caption{Throughput (left) and abort rate (right) under varying read interval}
    \label{chart:read_interval}

\vspace{-10pt}
\end{figure}

\textbf{Read Interval.}
To simulate the scenario where a transaction incurs heavy computations, we increase the interval between consecutive reads during the transaction execution.
When a transaction takes longer to execute, there is a higher probability to read across blocks.
{\fabricPlusplus} prevents this scenario by aborting transactions that read across blocks even though some may be serializable. 
This is evidenced by a larger proportion of transactions that are early aborted during the Execution phase, as shown in \Cref{chart:read_interval} (right).
{\foccs} consistently reads from a valid block snapshot and, hence, the effect of extending a transaction's execution results in a higher end-to-end latency.
This leads to more concurrent transactions which, in turn, results in a higher abort rate for both \textit{c-ww} and the dangerous pattern in {\foccs}.
Notably, the performance of the vanilla Fabric drops drastically with longer transaction execution.
We attribute this to the read-write lock used during the simulation and block commit which prevents parallelism.

\revision{
\subsection{The Performance of {\ffabricSharp}}

Next, we analyze the performance of our approach on top of FastFabric~\cite{gorenflo2019fastfabric}. 
FastFabric separates peers in the same administrative domain into endorser, storage, and validator nodes. 
Endorsers, storage nodes, and validators are solely responsible for the transaction execution, block persistence, and transactions validation, respectively.
These optimizations enable FastFabric to obtain a speedup of 6 compared to vanilla Fabric, as reported in ~\cite{gorenflo2019fastfabric}.

We implement our reordering techniques on top of FastFabric and name the resulting system {\ffabricSharp}.
In our experimental setup of four peers, we set two peers as endorsers, one as storage, and one as validator.
%
We use two workloads based on the original Smallbank to evaluate the effectiveness of our approach when the workloads exhibit less conflicts compared to the modified Smallbank used in the previous set of experiments.
The first workload consists of uniform update-only transactions (\textit{Create Account}), which are all serializable due to the absence of read operations and hence \textit{anti-rw}. 
The second workload is a mix of 50\% read-only transactions (\textit{Query Account}), 30\% transactions that modify a single account (\textit{Deposit Checking, Write Check, Transaction Saving}), and 20\% transactions that modify two accounts (\textit{Send Payment, Amalgamate}). 
The skewness of the accessed accounts is controlled by the zipfian parameter $\theta$. 

\begin{figure}[tp]

  \centering     
  \includegraphics[width=0.55\textwidth]{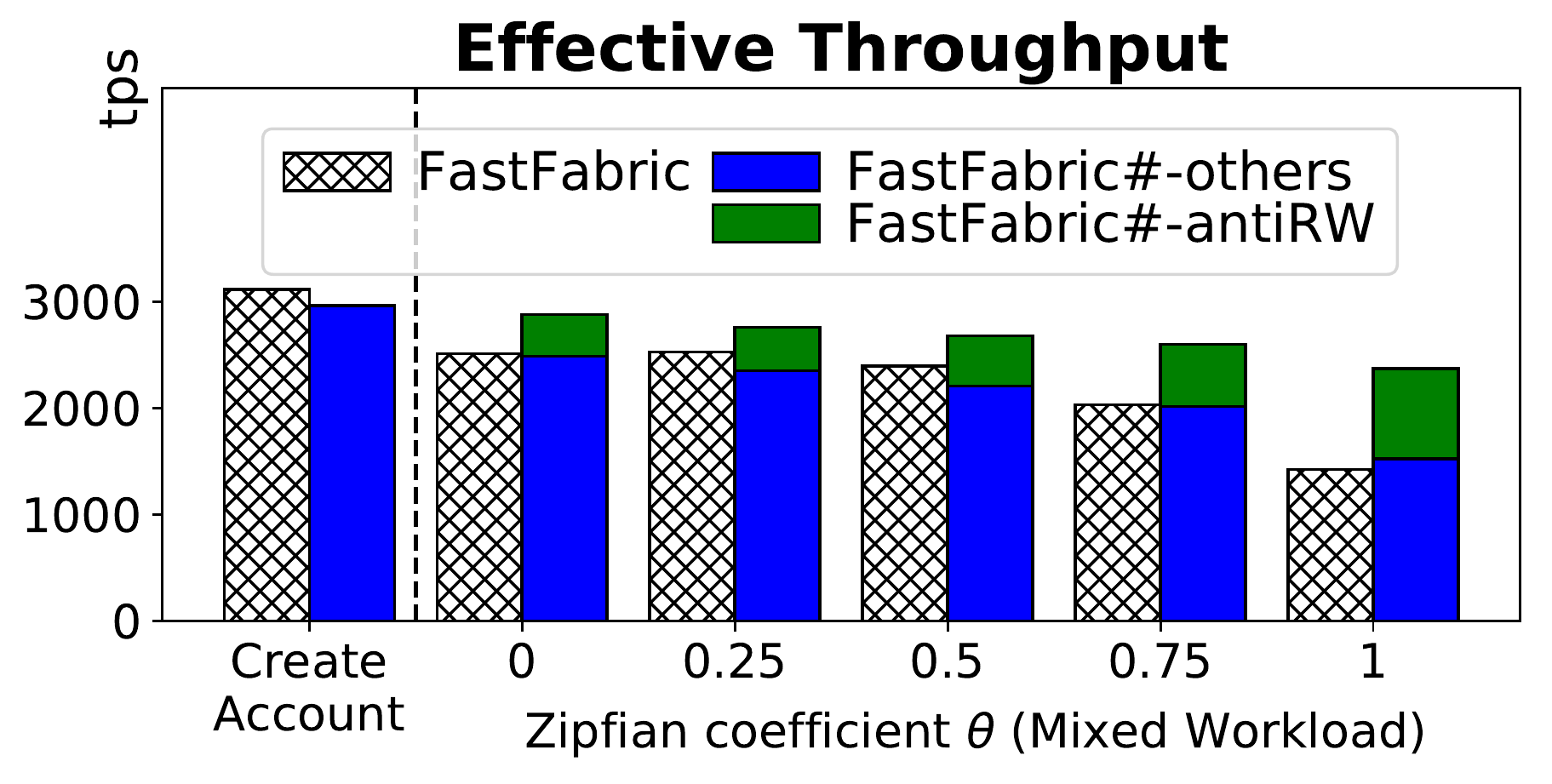}

  \caption{Effective throughput of FastFabric and {\ffabricSharp} under the contention-free (\textit{Create Account}) and mixed workload}
  \label{fig:fastfabricx}

	\vspace{-10pt}
\end{figure}

\Cref{fig:fastfabricx} reports the comparison between FastFabric and {\ffabricSharp}. 
With the contention-free \textit{Create Account} workload, FastFabric achieves a speedup of $4.5$ on our experimental setup compared to the original Fabric, namely 3114 tps versus 677 tps in terms of raw throughput.
The reordering overhead in {\ffabricSharp} is less than 5\% and its effective throughput achieves 2960 tps. 
Under the mixed workload, {\ffabricSharp} achieves higher throughput compared to FastFabric.
Moreover, the gap between {\ffabricSharp} and FastFabric grows with increasing skewness. 
When the $\theta=1$, {\ffabricSharp} can achieve 2370 tps, 66\% more than the 1424 tps of FastFabric. 
The gain of {\ffabricSharp} is mostly from the serialized transactions with \textit{anti-rw}, as highlighted in~\Cref{fig:fastfabricx}, which are all aborted by FastFabric.

We conclude that the benefits of our reordering outweigh the computation overhead, even under a heavy load. 
From previous experiments, this overhead is mostly due to reachability update, which depends on the transaction complexity. 
In practice, if our reordering becomes the bottleneck, we could simply adopt a discriminated approach.
For example, complex transactions with more referenced records can be checked with a simple serializability condition, e.g., the existence of \textit{anti-rw} conflicts.
On the other hand, simple transactions are passed to the fine-grained reordering. 
}



\section{Related Work}
\label{sec:related}

\textbf{Concurrency Control in Databases.}
Modern hardware with large memory and multi-core architecture opens up new
opportunities for redesigning OLTP RDBMSs and their concurrency control
\cite{kallman2008h}.
Various works have attempted to either make data access more cache-friendly
\cite{yu2018sundial, tu2013speedy} or extract more execution parallelism
\cite{yao2016exploiting, wu2016transaction}.
Others strove to streamline the transaction ordering based on workload
characteristics and application requirements \cite{guo2019adaptive,
  ding2018improving, wang2016mostly}.
Our work is closer to the latter approach, but we focus on blockchain systems.
In particular, we learn from the standardized transactional analysis techniques
in databases and propose a novel technique to efficiently reduce the number of
aborted transactions in EOV blockchains.

\textbf{Concurrency Control in Blockchains.}
The serial execution of smart contracts has long been the only solution adopted
by blockchains due to ease of reasoning.
Fabric, with its novel EOV architecture, is the pioneer system which formally
introduced concurrency management into blockchains.
It sparked a series of related optimizations, such as {\fabricPlusplus}
\cite{sharma2019blurring}, and enhanced architectures, like OXII
\cite{amiri2019parblockchain}.
{\fabricPlusplus} is the closest work to ours. 
It also aims to reduce the number of aborted transactions under EOV architecture.
However, our approach provides a more fine-grained concurrency control which can still serialize transactions aborted by {\fabricPlusplus}.

\textbf{Bridging the Gap between Database and Blockchain Transactions.}
Similarities between blockchains and databases have long been observed in their surveys~\cite{ruan2019blockchains, dinh2018untangling} and benchmarks~\cite{dinh2017blockbench}. 
There are several works addressing the atomicity of cross-chain or cross-shard
transactions in blockchains, by transitioning the well-established database
techniques, such as the classic Two Phase Commit \cite{herlihy2019cross,
  zakhary2019atomic, dang2019towards}.
Nathan et al. propose to implement blockchains on top of databases to
rely on their store procedures for richer contract expressiveness \cite{Nathan}.
Smart contract, as a key enabler for the transactional workload in blockchains,
is extensively optimized for better utility with lineage information
\cite{ruan2019fine} or confidentiality across applications
\cite{amiri2019caper}.

\section{Conclusions}
\label{sec:conclusion}

We propose a novel solution to efficiently reduce the transaction
abort rate in EOV blockchains by applying transactional analysis from
OCC databases.
We first draw a theoretical parallelism between EOV blockchains and
OCC databases.
Then, we propose a fine-grained concurrency control method and implement it in
{\fabricSharp} and {\ffabricSharp} based on Fabric and FastFabric, respectively.
\revision{Our experimental analysis shows that both {\fabricSharp} and {\ffabricSharp} outperform other blockchain systems, including the vanilla Fabric,
{\fabricPlusplus}, and FastFabric.
Unlike databases that achieve high throughput, the blockchains' limited throughput due to factors related to security opens up opportunities for precise transaction management.
}


\clearpage
\section{Acknowledgements}
This research is supported by Singapore Ministry of Education Academic Research Fund Tier 3 under MOE's official grant number MOE2017-T3-1-007.

\appendix

\bibliographystyle{abbrv}
\bibliography{main}

\end{document}